\newcommand{\cv}[1]{} %
\newcommand{\av}[1]{#1} %
\av{
	\newcommand{\isarxiv}{1}
}

\cv{
\documentclass[a4paper,USenglish,cleveref, autoref, thm-restate]{lipics-v2021}
\hideLIPIcs  %
\usepackage{wrapfig}
}

\av{\documentclass{article}
  \pdfoutput=1
  \usepackage[totalwidth=420pt,totalheight=625pt]{geometry}
\date{}
}
\usepackage{todonotes}
\usepackage{markdown}
\usepackage{tikz}
\usepackage{bussproofs}  
\usepackage{tabularray}
\usepackage{mathtools}
\usepackage{booktabs}
\usepackage{longtable}
\usepackage{xcolor}
\cv{\usepackage{arydshln}}
\usepackage{tabularx}
\newcommand{\inci}[2]{e_{#1,#2}}
\newcommand{\incip}[2]{e'_{#1,#2}}
\newcommand{\vari}[2]{a_{#1,#2}}
\newcommand{\cl}[2]{c_{#1,#2}}

\newcommand{\true}[1]{\alpha_{#1}}
\newcommand{\matr}[2]{A(#1,#2)}
\newcommand{\pos}[1]{\text{\normalfont pos}(#1)}
\newcommand{\nega}[1]{\text{\normalfont neg}(#1)}

\newcommand{\fixvar}[0]{\mathcal{V}}
\newcommand{\fixlit}[0]{\mathcal{L}}
\newcommand{\surplus}[0]{\mathcal{S}}

\newcommand{\plus}[0]{\color{black}{$+$}\color{black}}
\newcommand{\minus}[0]{\color{black}{$-$}\color{black}}
\newcommand{\elim}[2]{#1^{#2}}
\newcommand{\fwf}[1]{10.55776/#1}
\usepackage{ifthen}

\av{
\usepackage{hyperref}
\usepackage{amsmath}
\usepackage{amsthm}
\usepackage{cleveref}
\usepackage{caption}
\newtheorem{theorem}{Theorem}
\newtheorem{lemma}[theorem]{Lemma}

\newtheorem{corollary}[theorem]{Corollary}

\captionsetup{font=small}
\usepackage{wrapfig}
\usepackage{multirow}
}

\newcommand{\drule}[4][]{$F_{#2}=
  \ifthenelse{\equal{#3}{0}}{\{\nil\}}{F_{#3}}
  \oplus_{#1}
  \ifthenelse{\equal{#4}{0}}{\{\nil\}}{F_{#4}}$}

\newcommand{\Mdrule}[4][]{$M^1_{#2} = F_{#3} \oplus_{#1} F_{#4}$}

\newcommand{\Drule}[4][]{F_{#2}&=&
  \ifthenelse{\equal{#3}{0}}{\{\nil\}}{F_{#3}}
  & \oplus_{#1}&
  \ifthenelse{\equal{#4}{0}}{\{\nil\}}{F_{#4}}}

\title{
	Small unsatisfiable \texorpdfstring{$k$}{k}-CNFs with bounded literal occurrence\av{\footnote{This is the full version of a
          paper to appear in the proceedings of SAT 2024.}}
}

\av{
  \author{%
    Tianwei Zhang, Tom\'{a}\v{s} Peitl, and Stefan Szeider\\[4pt]
\small Algorithms and Complexity Group\\[-3pt]
\small TU Wien, Vienna, Austria\\[-3pt]
\small \texttt{\{zhangtw,peitl,sz\}@ac.tuwien.ac.at}}
}

\cv{
\author{Tianwei Zhang}{Algorithms and Complexity Group, TU Wien, Austria
\url{https://www.ac.tuwien.ac.at/people/zhangtw/}}{zhangtw@ac.tuwien.ac.at}{https://orcid.org/0009-0000-3745-5234}{}

\author{Tom\'{a}\v{s} Peitl}{Algorithms and Complexity Group, TU Wien, Austria
  \url{https://www.ac.tuwien.ac.at/people/peitl}}{peitl@ac.tuwien.ac.at}{https://orcid.org/0000-0001-7799-1568}{}

\author{Stefan Szeider}{Algorithms and Complexity Group, TU Wien, Austria \url{https://www.ac.tuwien.ac.at/people/szeider/}}{sz@ac.tuwien.ac.at}{https://orcid.org/0000-0001-8994-1656}{}

\authorrunning{T. Zhang, T. Peitl, and S. Szeider} %

\Copyright{ and Tianwei Zhang, Tom\'{a}\v{s} Peitl, and Stefan Szeider} %



\category{} %

\relatedversion{Full version is available on arXiv~\cite{ZhangPeitlSzeider24arxiv}} %

\supplement{The source code and scripts for reproducing the results are available at} %
\supplementdetails[]{Software, Dataset}{https://doi.org/10.5281/zenodo.11282310}

\funding{The project leading to this publication has received funding
  from the European Union's Horizon 2020 research and innovation
  programme under grant agreement No.\ 101034440 \marginpar{\includegraphics[width=10mm]{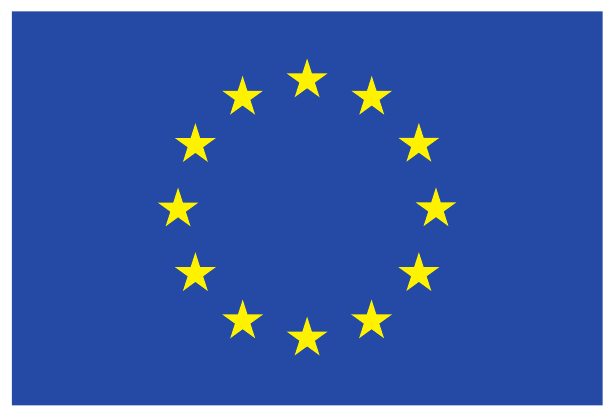}}, and was supported by
  the Austrian Science Fund (FWF) within the project
  \fwf{P36688}.
}

\nolinenumbers %

} %

\def\hy{\hbox{-}\nobreak\hskip0pt}

\newcommand{\SB}{\{\,}
\newcommand{\SM}{\;{:}\;}
\newcommand{\SE}{\,\}}

\newcommand{\CNF}[1]{#1\text{\normalfont-CNF}}
\newcommand{\MU}{\text{\normalfont MU}}

\newcommand{\nil}{\square}

\newcommand{\Card}[1]{|#1|}

\newcommand{\mtext}[1]{\text{\normalfont\sffamily #1}}

\newcommand{\var}{\mtext{var}}
\newcommand{\lit}{\mtext{lit}}
\newcommand{\ol}[1]{\overline{#1}}

\newcommand{\occ}[2]{#1[#2]}

\renewcommand{\subparagraph}[1]{\medskip\noindent\textbf{#1.}}

\NewDocumentEnvironment{formula}{ +b }{%
    \ensuremath{%
        \left\lgroup
        \begin{matrix}
            #1
        \end{matrix}
        \right\rgroup
    }
}{}

\definecolor{ijcaired}{HTML}{D22817}
\definecolor{azure}{rgb}{0.0, 0.5, 1.0}
\definecolor{applegreen}{rgb}{0.55, 0.71, 0.0}
\definecolor{auburn}{rgb}{0.43, 0.21, 0.1}
\definecolor{bittersweet}{rgb}{1.0, 0.44, 0.37}
\definecolor{byzantine}{rgb}{0.74, 0.2, 0.64}
\definecolor{darkmagenta}{rgb}{0.55, 0.0, 0.55}
\definecolor{darkelectricblue}{rgb}{0.33, 0.41, 0.47}

\newcommand{\pbound}{p}
\newcommand{\nbound}{q}

\newcommand{\GGG}{\mathcal{G}}

\newcommand{\n}[1]{{\overline{#1}}}
\newcommand{\ksnm}[4]{\Phi_{#1,#2}^{#3,#4}}
\newcommand{\kpqnm}[5]{\Phi_{#1,#2,#3}^{#4,#5}}
\newcommand{\satisfies}[2]{\Sigma^{#1,#2}}
\newcommand{\vdeg}[1]{{\deg^\var(#1)}}
\newcommand{\ldeg}[1]{{\deg(#1)}}
\newcommand{\vardeg}[2]{{\deg^\var_{#2}(#1)}}
\newcommand{\litdeg}[2]{{\deg_{#2}(#1)}}
\newcommand{\mks}[2]{{\mu(#1,#2)}}
\newcommand{\mkssc}[3]{{\mu(#1,#2,#3)}}
\newcommand{\mkst}[3]{{\mu(#1,#2,#3)}}
\newcommand{\mkstsc}[4]{{\mu(#1,#2,#3,#4)}}

\newcommand{\sv}[1]{}

\begin{document}
\maketitle
\thispagestyle{empty}
\begin{abstract}
We obtain the smallest unsatisfiable formulas in subclasses of $k$-CNF (exactly $k$ distinct literals per clause) with bounded variable or literal occurrences. 
Smaller unsatisfiable formulas of this type translate into stronger inapproximability results for MaxSAT in the considered formula class.
Our results cover subclasses of $3$-CNF and $4$-CNF; in all subclasses of $3$-CNF we considered we were able to determine the smallest size of an unsatisfiable formula; in the case of $4$-CNF with at most~5 occurrences per variable we decreased the size of the smallest known unsatisfiable formula.
Our methods combine theoretical arguments and symmetry-breaking exhaustive search based on SAT Modulo Symmetries~(SMS), a~recent framework for isomorph-free SAT-based graph generation. To this end, and as a standalone result of independent interest, we show how to encode formulas as graphs efficiently for SMS. 

\end{abstract}

\section{Introduction}
\label{sec:intro}

A $(k,s)$\hy formula is a propositional CNF formula in which each
clause has exactly $k$ distinct literals and each variable occurs
(positively or negatively) in at most $s$ clauses.  Since
Tovey~\cite{Tovey84} initiated the study of $(k,s)$-CNF formulas in
1984, they have been the subject of intensive investigation
\cite{BermanKarpinskiScott03,%
BermanKarpinskiScott03c,%
Dubois90,%
Gebauer12,%
GebauerSzaboTardos16,%
HoorySzeider05,%
HoorySzeider06,%
Jurenka11,%
KratochST93,%
SavickySgall00}.
Using Hall's Marriage Theorem, Tovey showed that all $(3,3)$-CNF formulas are satisfiable, but allowing a
fourth occurrence per variable yields a class of formulas for which the
satisfiability problem is NP-complete. Kratochv\'{\i}l, et al.\
\cite{KratochST93} generalized this result and showed that for each
$k\geq 3$, there exists a threshold $s=f(k)$ such that all
$(k,f(k))$-formulas are satisfiable and checking the satisfiability of
$(k,f(k)+1)$\hy formulas (the \emph{$(k,s)$-SAT problem}) is
NP-complete. Therefore, determining whether $(k,s)$-SAT is NP-hard boils down to identifying an unsatisfiable $(k,s)$\hy formula.
While tight asymptotic bounds for the threshold have been obtained~\cite{GebauerSzaboTardos16}, exact values are only known for
$k \leq 4$ \cite{HoorySzeider05}: $f(3)=3$, $f(4)=4$; $f(5)\in [5,7]$.
No decision procedure is known for determining the threshold $f(k)$,
since no upper bound on the size of a smallest unsatisfiable
$(k,s)$\hy formula is known for $k>4$.
 
It is an intriguing question of extremal combinatorics to determine
the size of the smallest unsatisfiable $(k,s)$\hy formulas (with
$s>f(k)$), and this paper sets out to address this question for various values of the parameters. This requires proving lower and upper bounds---an upper bound typically consists in exhibiting a formula with suitable parameters, while a lower bound requires a proof that no such formulas of a particular size exist.
For $k=3$, the unsatisfiable $(3,4)$-formula
constructed by Tovey~\cite{Tovey84} has 40 clauses,
this was later improved to 19 and then 16 by Berman, et al.\
\cite{BermanKarpinskiScott03,BermanKarpinskiScott03c}.  

A simple counting argument shows that such a formula must contain at least $8$ clauses, hence there is a significant gap between the
known lower and upper bounds.
For $k=4$, the gap is even larger.  St\v{r}\'{\i}brn\'{a}
\cite{Stribrna94} constructed an unsatisfiable (4,5)-CNF formula with 449 clauses, which Knuth \cite[p.~588]{Knuth23} improved to 257. The same counting argument shows that such a formula must contain at least $16$ clauses.

All constructions above have the following in common. A satisfiable $(k,s)$\hy
formula with a backbone variable $x$ that must be false in all satisfying truth assignments is first constructed. Such a formula is called a \emph{$(k,s)$\hy enforcer}.  Then one combines $k$ copies of this formula together with a $k$-clause with the $k$ backbone variables to obtain an unsatisfiable $(k,s)$-CNF formula. Aside from providing an upper bound for the size of a smallest unsatisfiable $(k,s)$-formula, the size of unsatisfiable $(k,s)$\hy enforcers has a direct effect on inapproximability results for certain NP-hard Max-SAT
problems~\cite{BermanKarpinskiScott03}. We also address the problem of the smallest size of a $(3,s)$ or $(3,p,q)$-enforcer in this paper.

\subsection*{Contribution} In this paper, we develop a general
approach to computing small unsatisfiable
$(k,s)$\hy formulas and $(k,s)$\hy enforcers. We also consider the
more fine-grained setting of $(k,\pbound,\nbound)$-formulas, where
$\pbound$ and~$\nbound$ bound the number of positive and negative
occurrences per variable, respectively.  We observe a similar
threshold phenomenon in the complexity of $(k,\pbound,\nbound)$-SAT as
in the case of $(k,s)$-SAT (\Cref{lem:unsat-np}). 
 
Our approach rests on utilizing the SAT Modulo Symmetries (SMS)
framework \cite{KirchwegerSzeider24,KirchwegerPeitlSzeider23} for
isomorph-free generation of unsatisfiable $(k,s)$ and
$(k,p,q)$\hy formulas for a fixed number $n$ of variables and number
$m$ of clauses with parameters $k,s,p,q$.
 
The basic setting without any techniques for speed-up and divide-and-conquer scales up to about $m=13$.
We then use further theoretical arguments together with techniques for speed-up to determine the size of smallest formulas in various setting. 

\subparagraph{Smallest unsatisfiable formulas} We determined
the size of smallest unsatisfiable $(3,s)$ and $(3,p,q)$-formulas for all possible $s$, $p$ and $q$ the values are given in \Cref{tab:3-s} on the left for $(3,s)$-formulas and on the right for $(3,p,q)$-formulas.
\begin{table}[tbh]
  \caption{Size of smallest unsatisfiable $(3,s)$\hy formulas (left)
    and size of smallest unsatisfiable $(3,p,q)$\hy formulas (right).
    $\infty$ indicates that for these parameters no unsatisfiable
    formula exists. }
  \av{\medskip}
    \label{tab:3-s}
    \centering
      \setlength\tabcolsep{2.2mm}
      \av{\qquad}\begin{tabular}{@{}cccc@{}}
        \toprule
        $s\leq 3$ & $s=4$ & $s=5$ & $s\geq 6$ \\
        \midrule
        $\infty$ & 16 & 11 & 8\\ 
        \bottomrule
        \phantom{1,2,3} \\
        \phantom{1,2,3} \\
      \end{tabular}\hfill%
  \begin{tabular}{@{}cccccc@{}}
    \toprule
    &$q=1$ & $q=2$ & $q=3$ & $q=4$ & $q\geq 5$ \\
    \midrule
    $p=1$& $\infty$ & $\infty$ & 22 & 19 & 16\\
    $p=2$ &-& 20 & 11 & 10 & 10\\
    $p\geq 3$ & - & - & 8& 8 & 8\\
    \bottomrule
  \end{tabular}\av{\qquad}
    \end{table}

In particular, we identified a smallest unsatisfiable $(3,1,3)$\hy formula
with 22 clauses, which implies that $(3,1,3)$-SAT is NP-complete.
Hence, we have now a direct and streamlined proof for the 
dichotomy of $(3,p,q)$-SAT (\Cref{the:dich}).

\subparagraph{Smallest enforcers}
\Cref{tab:3-enf} summarizes our results on the size of smallest
$(3,s)$\hy enforcers and $(3,p,q)$\hy enforcers.
\begin{table}[tbh]
  \caption{Size of smallest $(3,s)$\hy enforcers (left)
    and smallest  $(3,p,q)$\hy enforcers (right).}
   \label{tab:3-enf}
\av{\medskip}
   \centering
  \setlength\tabcolsep{3mm} 
\av{\qquad}\begin{tabular}{@{}ccc@{}}
    \toprule
    $s\leq 3$ & $s=4$ & $s\geq 5$ \\
    \midrule
    $\infty$ & 5 & 4\\
    \bottomrule
    \phantom{1,2,3} & & \\
    \phantom{1,2,3} & & \\
  \end{tabular}\hfill%
  \begin{tabular}{@{}cccccc@{}}
    \toprule
    &$q=1$ & $q=2$ & $q=3$ & $q=4$ & $q\geq5$ \\
    \midrule
    $p=1$  & $\infty$ & $\infty$ & 7 & 6 & 5\\
    $p=2$ & - & 10 & 5 & 5 & 5\\
    $p\geq 3$ & - & - & 5 & 5 & 5 \\
    \bottomrule
  \end{tabular}\av{\qquad}
 \end{table}

\subparagraph{A smaller unsatisfiable $(4,5)$-formula}
Recall that $f(4)=5$. The smallest known unsatisfiable
$(4,5)$\hy formula is due to Knuth and has 257 clauses. We improve
this by exhibiting an unsatisfiable $(4,5)$\hy formula with 235
clauses. We obtain this by first computing an auxiliary formula with SMS and then constructing from it an unsatisfiable $(4,5)$\hy formula by disjunctive splitting.

\subparagraph{Structure of the paper}
After this introduction, and preliminaries, the paper is organized into two main parts.
In Section~\ref{sec:compute}, we lay out the principle encoding that provides the basis for all our results. 
In Section~\ref{sec:advanced}, we dive into the more complicated, technical aspects that are
necessary to rule out the existence of larger formulas and obtain
better lower bounds. \cv{Lemmas and Theorems marked with $\star$ have proofs in the full version~\cite{ZhangPeitlSzeider24arxiv}.}

\av{
\subparagraph{Source Code}
The source code and scripts for reproducing the results are available
at Zenodo~\cite{ZhangPeitlSzeiderZenodo24}.
}

\section{Preliminaries}
\label{sec:prelim}
For positive integers~$k<\ell$, we write $[k] = \{1,2,\dots,k\}$, $[-k]
= \{-k, -(k-1), \dots, -1\}$, and $[k,\ell]=\{k,\dots,\ell\}$. 
We assume familiarity with fundamental notions of propositional logic~\cite{KleineBuningLettman99}. 
In this paper we will talk about \emph{(minimally) unsatisfiable} propositional formulas represented as graphs, specified by properties expressed as \emph{quantified Boolean formulas}, and about \emph{isomorphisms (symmetries)} of these formulas.
We review the relevant basics below.

\subparagraph{CNF formulas} 
\label{par:prelims-cnf}
A \emph{literal} is a (propositional) variable $x$ or a negated
variable $\ol{x}$, whereby $\n{\n{x}} = x$. We write $\var(x) := \var(\n{x}) := x$ for the
variable belonging to a literal.
A set $S$ of literals is \emph{tautological} if
$S\cap \ol{S}\neq \emptyset$, where
$\ol{S}=\SB \ol{x} \SM x\in S\SE$.  A \emph{clause} is a finite
non-tautological set of literals.  A $k$-clause is a clause that
contains exactly $k$ literals. The $0$-clause is denoted by $\nil$.
A \emph{(CNF) formula} is a finite set of clauses.
For $k\geq 1$, a \emph{$k$-CNF formula} is a formula in
which all clauses are $k$\hy clauses (please note that some authors
allow a $k$-CNF formula to contain clauses with fewer than $k$
literals, but it is significant in our context that the number is
exactly $k$) and a \emph{$(\leq k)$-CNF formula} is a formula in
which all clauses contain at most $k$ literals.
A variable~$x$ \emph{occurs positively} in a clause~$C$
if $x \in C$, it \emph{occurs negatively} in $C$ if $\ol{x}\in C$, and it \emph{occurs} in $C$ if it occurs in $C$ positively or negatively.
For a literal $x$, $\occ{F}{x}$ denotes the set of clauses in $F$ in which $\var(x)$ occurs.
We will often write a $\leq k$-CNF with $m$ clauses as a $k \times m$ matrix whose columns are the clauses, and entries are literal occurrences. When a clause has $r < k$ literals, we write $\times$ in the last $k - r$ rows in the corresponding column of the matrix. 

\subparagraph{Counting occurrences} For a clause $C$, we write
$\var(C)$ for the set of variables that occur in $C$, and for a CNF
formula $F$ we write $\var(F)=\bigcup_{C \in F} \var(C)$. The
\emph{degree} of a variable in a formula~$F$ is defined as
$\vardeg{x}{F} := |\occ{F}{x}|$.  For literals, we put
$\vardeg{\n{x}}{F} := \vardeg{x}{F}$.  The \emph{degree} of a literal
$x$, denoted by $\litdeg{x}{F}$, is the number of clauses in which the
literal occurs.  We may omit the subscript $F$ when it is clear from
the context.

For $k,s\geq 1$, a \emph{$(k,s)$-formula}
is a $k$-formula in which each variable occurs in at most $s$
clauses, and a \emph{$(\leq k,s)$-formula}
is a $(\leq k)$-formula in which each variable occurs in  at most $s$
clauses. For $k,p,q\geq 1$, a \emph{$(k,p,q)$-formula}
is a $k$-formula in which each variable occurs in at most $p$
clauses positively and in at most $q$ clauses negatively, and a \emph{$(\leq k,p,q)$-formula} is a $(\leq k)$-formula with the same constraint. 
Without loss of generality, we will assume $p \leq q$ for $(k, p,
q)$-formulas as we can always  swap positive and negative literals.

We define $\mks{k}{s}$ to be the number of clauses of a
smallest unsatisfiable $(k, s)$-formula, and $\mkst{k}{p}{q}$ denotes
that of the smallest unsatisfiable $(k,p,q)$-formula.

\subparagraph{Bounded literal occurrence SAT} A \emph{truth
  assignment} for a set $X$ of variables is a mapping
$\tau:X\rightarrow \{0,1\}$. In order to define $\tau$ on literals, we
set $\tau(\ol{x})=1-\tau(x)$.  A truth assignment~$\tau$ {\it
  satisfies} a clause $C$ if $C$ contains at least one literal~$x$
with $\tau(x)=1$, and $\tau$ \emph{satisfies} a formula~$F$ if it
satisfies every clause of~$F$. In the latter case, we call $F$
\emph{satisfiable}.  The {\sc Satisfiability} problem (SAT) is to
decide whether a given formula is satisfiable.  $(k,s)$-SAT is SAT
restricted to $(k,s)$-formulas, and $(k,p,q)$-SAT is SAT restricted to
$(k,p,q)$-formulas.

\subparagraph{Enforcers}
A $(k,s)$-\emph{enforcer} is a satisfiable $(k,s)$-formula $F$ with a variable $x$ with $\vardeg{x}{F} < s$ that is set to the same value in every satisfying assignment.
A~$(k,p,q)$-\emph{enforcer} is a satisfiable $(k,p,q)$ formula $F$ with a variable $x$ which is either set to true in every satisfying assignment, and then $\litdeg{\n{x}}{F} < q$, or it is set to false in every satisfying assignment, and then $\litdeg{x}{F} < p$.
We say the literal of $x$ that is set to true in every satisfying assignment is \emph{enforced}.
An enforcer can be completed into an unsatisfiable $(\leq k,s)$ or $(\leq k,p,q)$-formula by adding the unit clause containing the negation of the enforced literal.

\subparagraph{Minimal unsatisfiability}
A CNF formula is \emph{minimally unsatisfiable} if it is unsatisfiable
but dropping any of its clauses results in a satisfiable formula.
Let $\MU$ denote the class of all minimally unsatisfiable formulas.
The \emph{deficiency} of a CNF formula $F$ is
$\delta(F)=\Card{C}-\Card{\var(F)}$.
It is known that $\delta(F)>0$ for any $F\in \MU$~\cite{AharoniLinial86}; therefore it is natural to parameterize $\MU$
by deficiency and to consider the classes $\MU(d):=\SB F\in \MU \SM
\delta(F)=d \SE$ for $d\geq 1$.

\subparagraph{Variable elimination}
It is well-known that one can eliminate variables of a CNF formula by a process often called \emph{DP-resolution}, after an algorithm of Davis and Putnam~\cite{DavisPutnam60}, as follows.
For two clauses $C, D$ with $x \in C, \n{x} \in D$, the \emph{resolution} rule yields the \emph{resolvent} clause $C \cup D \setminus \{x, \n{x}\}$.
Let $F$ be a CNF and $x \in \var(F)$.
We define $\elim{F}{x} := F \setminus \occ{F}{x} \cup \{C \cup D \setminus \{x,\bar{x}\} \,|\, C, D \in F ; C \cap \n{D} = \{x\}\}$.
In other words, the result of eliminating $x$ from $F$ is the formula that contains all clauses where $x$ does not occur together with all possible non-tautological resolvents on $x$.
It is easy to see that $\exists x \; F$ and $\elim{F}{x}$ are logically equivalent, and in particular, if $F$ is unsatisfiable, so is $\elim{F}{x}$.

\subparagraph{Blocked clauses}
\label{par:blocked-clauses}
A clause $C$ in a CNF $F$ is \emph{blocked in $F$ on the literal} $x \in C$ if for every $C' \in F$ with $\n{x} \in C'$, there exists a variable $y \neq \var(x)$ with $y \in C, \n{y} \in C'$ or $y \in C', \n{y} \in C$.
A clause is \emph{blocked in $F$} if it is blocked on at least one of its literals.
Blocked clauses are a fundamental SAT preprocessing technique: when $C$ is blocked in $F$ and $F$ is satisfiable, then $F \cup \{C\}$ is also satisfiable~\cite{JarvisaloBiereHeule10,Kullmann99r}; in other words, blocked clauses may be added or removed without impacting satisfiability (notice that this also follows from soundness of variable elimination).
We will use the simple corollary that a minimally unsatisfiable formula cannot contain a blocked clause.

\subparagraph{QBF}
\label{subsec:qbf}
Quantified Boolean formulas generalize propositional logic with quantification.
In this paper, we will need only the fragment of \emph{closed prenex $2$-QBFs} with one quantifier alternation.
A $2$-QBF has the form $\exists X \forall Y \Phi(X, Y)$, where $X$ and $Y$ are sets of propositional variables, and $\Phi$ is a propositional formula.
A $2$-QBF is \emph{true} if there exists an assignment $\tau : X \to \{0, 1\}$ such that $\Phi(\tau(X), Y)$ evaluates to true for every assignment to $Y$, where $\tau(X)$ denotes the substitution of $\tau$ values for $X$ into $\Phi$.
The formula is \emph{false} if no such assignment $\tau$ exists.

\subparagraph{Graphs}
\label{subsec:graphs}
We only use undirected and simple graphs (i.e.,
without parallel edges or self-loops). A \emph{graph} $G$ consists of
set $V(G)$ of vertices and a set $E(G)$ of edges; we denote the edge
between vertices $u,v\in V(G)$ by $uv$ or equivalently $vu$.

We write $\GGG_n$ to denote the class of all graphs with $V(G) = [n]$.  The \emph{adjacency matrix} $A$ of a
graph $G \in \GGG_n$ is the $n\times n$ $\{0,1\}$-matrix where the element at row $v$ and column $u$, denoted by $\matr{v}{u}$, is $1$ iff $vu \in E(G)$.

\subparagraph{Isomorphisms}
\label{subsec:isomorphisms}
For a permutation $\pi : [n] \rightarrow [n]$, $\pi(G)$ denotes the graph obtained from $G\in \GGG_n$ by the
permutation $\pi$, where $V(\pi(G)) = V(G) = [n]$ and
$E(\pi(G))=\SB \pi(u)\pi(v)\SM uv \in E(G) \SE$.
Two graphs $G_1,G_2\in \GGG_n$ are \emph{isomorphic} if there is a
permutation $\pi : [n] \rightarrow [n]$ such that $\pi(G_1)=G_2$; in this case $G_2$ is an \emph{isomorphic copy} of $G_1$.
A \emph{partially defined graph}~\cite{KirchwegerSzeider21} is a graph $G$ where
$E(G)$ is split into two disjoint sets~$D(G)$ and~$U(G)$.  $D(G)$
contains the \emph{defined} edges, $U(G)$ contains the \emph{undefined} edges.  A
(\emph{fully defined}) graph is a partially defined graph $G$ with
$U(G)=\emptyset$.
A partially defined graph $G$ can be \emph{extended} to a
graph $H$ if  $D(G) \subseteq E(H) \subseteq D(G) \cup
U(G)$.

\subparagraph{CNF Formulas as graphs}
\label{subsec:2-graphs}
For sets $S, S', T$, we write  $T = S \uplus S'$ if $T = S \cup S'$ and $S \cap S' = \emptyset$.
A \emph{2-graph} is an undirected graph $G=(V,E)$ together with a
partition of its vertex set into two disjoint \emph{blocks} $V_1 \uplus V_2 = V$.  Two
2-graphs $G=(V_1 \uplus V_2,E)$ and $G'=(V_1' \uplus V_2',E')$ are
\emph{isomorphic} if there exists a bijection
$\phi: V_1 \uplus V_2 \rightarrow V_1' \uplus V_2'$ such that
$v\in V_i$ if and only if $\phi(v)\in V_i'$, $i=1,2$, and
$\{u,v\} \in E$ if and only if $\{\phi(u),\phi(v)\} \in E'$.
The \emph{clause-literal graph} of a CNF formula $F$ is the 2-graph $G(F)=(V_1 \uplus V_2,E)$ with $V_1=\lit(F)$, $V_2= F$, and 
$E=\SB \{x,\ol{x}\} \SM x\in \var(F)\SE \cup \SB \{C,\ell\} \SM C\in
F, \ell \in C \SE$.
We refer to the edges $\SB x, \n{x} \SE$ as \emph{variable} edges.
It is easy to verify that any two CNF formulas are isomorphic if and only
if their clause-literal graphs are isomorphic (as 2-graphs). Note that we can safely assume that the first $|V_1|$ rows of the adjacency matrix of a clause-literal graph correspond to the vertices in $V_1$, and we will thoroughly do so throughout this paper.

\subparagraph{SAT Modulo Symmetries (SMS)}
\label{subsec:sms}
SMS \cite{KirchwegerSzeider24} is a framework that augments a CDCL (conflict-driven clause learning) SAT solver~\cite{FichteHLS23,MarquessilvaLynceMalik09} with a custom propagator that can reason about symmetries, allowing to search modulo isomorphism for graphs in $\GGG_n$ satisfying a property specified in (quantified) propositional logic. 

During search, the SMS propagator can trigger additional conflicts on
top of ordinary CDCL and consequently learn \emph{symmetry-breaking clauses}, which exclude isomorphic copies of graphs. More precisely,
only those copies are kept which are lexicographically minimal (\emph{canonical}) when
considering the rows of the adjacency matrix concatenated into a single vector. A key
component is a minimality check, which decides whether a partially
defined graph can be extended to a minimal graph; if it cannot, a
corresponding clause is learned.  For a full description of SMS, we
refer to the original work where the framework was introduced~\cite{KirchwegerSzeider24}.
In SMS, it is possible to specify a partition of the vertex set
and restrict the symmetry breaking to those permutations
that preserve the partition.
In Section~\ref{subsec:formulas-sms}, we explain how we can specify partitions to efficiently generate formulas (represented by 2-graphs) modulo isomorphism with SMS.

\section{Basic encoding}
\label{sec:compute}

In this section we explain the methodology all of our investigations build upon, and review results already obtainable with it. In the next section, we will delve into technical details and improvements that are necessary to scale up this basic approach.

The idea is to first list all possible number of clauses a smallest $(k,s)$ or $(k,p,q)$-formula can have, and then decide whether an unsatisfiable formula exists with increasing $m$. For each $m$, we split the decision problem further by specifying number of variables $n$ the sought formula has.
For fixed $m$ and $n$, we reduce this task to deciding the satisfiability of a suitable quantified Boolean formula and give it to QBF-enabled SMS. Because we gradually increase $m$, the first time we hit a satisfiable instance we know the formula thus produced is smallest possible. The desired QBF should have its models correspond to unsatisfiable $(k, s)$-formulas with $n$ variables and $m$ clauses. Since unsatisfiability is coNP-complete, we cannot hope to obtain a polynomial-size propositional encoding (unless $\text{NP}=\text{coNP}$), and instead we use a $2$-QBF of the form $\exists X \forall Y \ksnm{k}{s}{n}{m}(X) \land \neg \satisfies{n}{m}(X,Y)$ (or $\exists X \forall Y \kpqnm{k}{p}{q}{n}{m}(X) \land \neg \satisfies{n}{m}(X,Y)$), where $\ksnm{k}{s}{n}{m}(X)$ ($\kpqnm{k}{p}{q}{n}{m}(X)$) expresses that $X$ represents a $(k, s)$-formula ($(k,p,q)$-formula) with $n$ variables and $m$ clauses, and $\satisfies{n}{m}(X,Y)$ expresses that the assignment represented by $Y$ satisfies $X$. 

\subsection{Hard-coding the first part of the clause-literal graph}
\label{subsec:formulas-sms}
Before we delve into the details of the encoding, we observe the following fact about the lexicographically minimal matrix of a clause-literal graph. 

Let $\pos{i} = r$ if the $r$th row/column of the adjacency matrix represents the positive literal of the $i$th variable, and similarly for $\nega{i}$. The first block of a clause-literal graph contains the vertices corresponding to literals, i.e., we have $\pos{i},\nega{i}\in[2n]$ for all $i\in[n]$.
\begin{theorem}
	\label{thm:matching-canonical}
	The lexicographically minimal matrix of any clause-literal graph is antidiagonal in the upper-left (variables) block, i.e., for all $i, j \leq 2n$, $\matr{i}{j} = 1$ iff $i + j = 2n+1$.
	In other words, for each $i$ we have $\{\pos{i}, \nega{i}\} = \{j, 2n+1-j\}$ for some $j \in [n]$.
\end{theorem}
\begin{proof}
	Towards a contradiction, consider a lexicographically minimal adjacency matrix of some clause-literal graph, and the row $i$, where antidiagonality is first violated, because the $1$-entry is in column $2n+1-j$ for $j>i$ (and not $j=i$ as it should be; notice that $j<i$ is impossible since the literal $2n+1-j$ would have to be adjacent to both $j$ and $i$).
	Then, swapping the vertices $2n+1-j$ and $2n+1-i$ yields a lexicographically smaller matrix.
\end{proof}

\begin{figure}
    \centering
    \begin{tikzpicture}[scale=.8]
\tikzstyle{var}=[fill=lightgray, draw=black, shape=circle, inner
sep=0pt, minimum size=4pt]
\tikzstyle{cls}=[fill=white, draw=black, shape=rectangle, inner
sep=0pt, minimum size=4pt]
  \draw 
(0,0)    node[var, label= $x$] (x) {} 
(1,0)    node[var, label={$\ol{x}$}] (-x) {} 
(2,0)    node[var, label={$y$}] (y) {} 
(3,0)    node[var, label={$\ol{y}$}] (-y) {} 
(0,-1)    node[cls, label=below:{$C_1$}] (C) {} 
(1.5,-1)    node[cls, label=below:{$C_2$}] (D) {} 
(3,-1)    node[cls, label=below:{$C_3$}] (E) {} 
(x)--(-x) (y)--(-y)
(x)--(C)--(y) (x)--(D)--(-y) (-x)--(E) 
;

\draw(9,-0.5) node (matrix) {$\begin{array}{@{}c|cccc|@{\ }c@{\ }c@{\ }c@{\ }}
               & \overline{x} & \overline{y} & y & x & C_1 & C_2 & C_3\\
                    \hline
				\overline{x} & 0 & 0 & 0 & 1 & 0 & 0 & 1 \\
				\overline{y} & &0 & 1 & 0 & 0 & 1 & 0 \\
                y & &  & 0 & 0 & 1 & 0 & 0 \\
                x & &  &   & 0 & 1 & 1 & 0 \\
                \hline
                C_1 & & & & &0 &0 & 0\\
                C_2 & & & & & & 0 & 0\\
                C_3 & & & & & & & 0\\
                \hline
		\end{array} $};
          \end{tikzpicture}
    \caption{Consider the formula $F=\{C_1,C_2,C_3\}$ with $C_1=\{x,y\}$, $C_2=\{x,\ol{y}\}$, and $C_3=\{\ol{x}\}$. We see the corresponding 2-graph and  the upper part of its lexicographically minimal adjacency matrix. Observe how the left part is indeed antidiagonal.}
    \label{fig:example}
\end{figure}
Theorem~\ref{thm:matching-canonical} shows that we can hard-code the top-left $2n\times2n$ matrix of the adjacency matrix. Doing so has the following benefits. The immediate advantage is that with fewer undecided variables to solve, SMS terminates more quickly. Also, this breaks the symmetries from reordering the literal nodes of the clause-literal matrix. Finally, as we will see more clearly in the next part, fixing the matching between literal nodes reduces the size of the encoding since we no longer need to describe the cardinality constraints on how many times a variable can occur \emph{conditionally} on an undecided matching among literal nodes. \Cref{fig:example} shows an example of a clause-literal graph and its corresponding lexicographically minimal adjacency matrix.

\subsection{Encodings for formulas}\label{sec:encoding}

In this part, we describe the detailed construction of the specific 2-QBF we use, whose models correspond to
unsatisfiable $(k, s)$-formulas with $n$ variables and $m$ clauses. We write our encoding in the standard circuit-QBF QCIR format~\cite{JordanKlieberSeidl16}.

We express cardinality constraints using cardinality networks~\cite{AsinNieuwenhuisOliverasRodriguez09}.
A cardinality network $\text{Card}^{a}_{b}$ takes
$b$ binary inputs and outputs the most significant $a$ inputs ordered
from more significant to less. For the purpose of illustration,
suppose we want variable $z$ to be true if and only if  exactly $d$ variables out of $x_1,\dots,x_b$ are true.
Let  $(y_1,\dots,y_{d},y_{d+1}):=\text{Card}^{d+1}_{b}(x_1,\dots,x_{b})$. Note that this way $y_i$ is true if and only if there are at least $i$ true inputs among $x_1,\dots,x_b$. So we can define $z:=y_d\land \lnot y_{d+1}$. When the number of inputs and outputs is clear from the context, we just write Card. Let $i,i'\in[n], j<j'\in[m]$, we define $\vari{i}{i'}:=\matr{\pos{i}}{\pos{i'}}$, $\vari{i}{-i'}:=\matr{\pos{i}}{\nega{i'}}$ and $\vari{-i}{-i'}:=\matr{\nega{i}}{\nega{i'}}$ for edges between literal vertices, $\inci{i}{j}:=\matr{\pos{i}}{j + 2n}$ and $\inci{i}{j}:=\matr{\nega{i}}{j + 2n}$ for edges from a literal node to a clause node, and $\cl{j}{j'}:=\matr{j+2n}{j'+2n}$ for edges between clause nodes. Take a sufficiently large~$x$. For all $i\in[n], t\in[n]\cup[-n]$ and $j\in[m]$, define
\[\begin{aligned}
(\omega_{t,1},\omega_{t,2},\dots,\omega_{t,x})&:=\text{Card}(\inci{t}{1},\inci{t}{2},\dots,\inci{t}{m}),\\
(\sigma_{i,1},\sigma_{i,2},\dots,\sigma_{i,x})&:=\text{Card}(\inci{i}{1},\inci{i}{2},\dots,\inci{i}{m},\inci{-i}{1},\inci{-i}{2},\dots,\inci{-i}{m}),\text{ and}\\
(\tau_{j,1},\tau_{j,2},\dots,\tau_{j,k+1})&:=\text{Card}(\inci{1}{j},\inci{2}{j},\dots,\inci{n}{j},\inci{-1}{j},\inci{-2}{j},\dots,\inci{-n}{j}).
\end{aligned}\]
The following are some useful properties expressed in propositional logic which we use as components in the desired 2-QBF. $F_1$ expresses that there are no edges between two clause nodes. $F_2$ expresses that each literal occurs at least once. $F_3$ expresses that a literal and its negation cannot occur in the same clause. $F^v_4$ expresses that each variable occurs at most $s$ times. $F_5$ expresses that each clause contains exactly $k$ literals. 
\[\begin{aligned}F_1&:=\bigwedge_{0\leq j < j' < m}\lnot \cl{j}{j'},\quad F_2:=\bigwedge_{i\in [n]\cup[-n]}\bigvee_{j\in[m]}\inci{i}{j},\quad F_3:=\bigwedge_{i\in[n],j\in[m]}\lnot\inci{i}{j}\lor\lnot\inci{-i}{j}.\\
F_4^v&:=\bigwedge_{i\in[n]}\lnot\sigma_{i,s+1},\quad F_4^l:=\bigwedge_{i\in[n]\cup[-n]}\lnot\omega_{i,q+1}\land\bigwedge_{i\in[n]}\lnot\omega_{i,p+1} \lor\lnot\omega_{-i,p+1},\\
F_5&:=\bigwedge_{j\in[m]}\tau_{j,k}\land\lnot\tau_{j,k+1}. \end{aligned}\]
It is easy to see that a smallest $(k,s)$ or $(k,p,q)$-formula must be minimally unsatisfiable, and thus we can also require (in $\ksnm{k}{s}{n}{m}$) that it contain no blocked clauses. Given $i\in[n]\cup[-n]$ and  $j,j'\in[n]$, define $\psi_{i,j,j'}$ and $\varphi_{i,j}$ as follows. Here $\varphi_{i,j}$ expresses that the $j$-th clause is not blocked on the literal $i$.
\[\psi_{i,j,j'}:=\bigwedge_{\substack{i'\in[n]\cup[-n]\\i'\not=i,-i}}(\lnot\inci{i'}{j}\lor\lnot\inci{-i'}{j'}),\quad \varphi_{i,j}:=\lnot\inci{i}{j}\lor\bigvee_{\substack{j'\in[m]\\j'\not=j}}(\inci{-i}{j'}\land\psi_{i,j,j'}).\]
$F_6$ expresses that the formula contained no blocked clauses.
\[F_6:=\bigwedge_{\substack{i\in[n]\cup[-n],\\j\in[m]}}\varphi_{i,j}\]
$F_7$ hard-codes the matching between the literal nodes. 
\[F_7:=\bigwedge_{i\in[n]}\vari{i}{-i}\land\bigwedge_{\substack{i\in[n],i'\in[-n]\\i\not=-i'}}\lnot\vari{i}{i'}\land\bigwedge_{i<i'\in[n]}\lnot\vari{i}{i'}\land\bigwedge_{i'<i\in[-n]}\lnot\vari{i}{i'}\]
Let $X:=\SB \matr{i}{j} \SM i<j\in[2n+m] \SE$ and $Y=\SB \true{i} \SM i\in [m] \SE$. For $i\in[n]$, we put $\true{-i}:=\lnot\true{i}$. Finally, define $\ksnm{k}{s}{n}{m}(X):=F_1\land F_2\land F_3 \land F^v_4 \land F_5 \land F_6\land F_7$, $\kpqnm{k}{p}{q}{n}{m}(X):=F_1\land F_2\land F_3 \land F^l_4 \land F_5 \land F_6\land F_7$, and \[
\satisfies{n}{m}(X,Y):=\bigwedge_{j\in[m]}\bigvee_{i\in[n]\cup[-n]}\true{i}\land\inci{i}{j}.\]

\subsection{Preliminary findings}
\label{subsec:bounds}

To determine the exact value of $\mks{3}{s}$ and $\mkst{3}{p}{q}$ for various choices of $s$, $p$ and $q$, we enumerate all permissible values of $(m,n)$, where $m$ is the number of clauses and $n$ is the number of variables, from smaller to bigger in the lexicographical order.
For each pair $(m, n)$, we solve the formula described in Section~\ref{sec:encoding} with SMS.%
\footnote{
	\url{https://github.com/markirch/sat-modulo-symmetries}
}
We terminate the solver if it cannot answer within $5$ days.

We ran the solver on a Sun Grid Engine (SGE) cluster consisting of heterogeneous machines running Ubuntu 18.04.6 LTS.%
\footnote{
	The cluster contains nodes with the following architectures: 2× Intel Xeon E5540 with 2.53 GHz Quad Core, 2× Intel Xeon E5649 with 2.53 GHz 6-core, 2× Intel Xeon E5-2630 v2 with 2.60GHz 6-core, 2× Intel Xeon E5-2640 v4 with 2.40GHz 10-core and 2× AMD EPYC 7402 with 2.80GHz 24-core.
}

We begin by observing that since allowing more occurrences yields a larger class of formulas, by definition, $\mks{k}{s} \leq \mks{k}{s-1}$, $\mkst{k}{p}{q} \leq \min \big( \mkst{k}{p-1}{q}, \mkst{k}{p}{q-1} \big)$, and  $\mks{k}{p+q} \leq \mkst{k}{p}{q}$.
The following lemmas provide preliminary bounds on $m$ and $n$ for our enumeration.

\begin{lemma}
	\label{lemma:model-bound-asym}
	An unsatisfiable $k$-CNF formula that contains a variable of type $(p, q)$ has at least $2^k + | q - p | $ clauses.
\end{lemma}
\begin{proof}
	Let $F$ be an unsatisfiable $k$-CNF with $m$ clauses, $n$ variables, and $x$ a variable of type $(p, q)$.
	$F|_{\n{x}} := \SB C \setminus \{x\} \in F \SM \n{x} \not \in C \SE$, obtained from $F$ by setting $x$ to false, is unsatisfiable, has $n-1$ variables, $p$ clauses of size $k-1$, and $m-p-q$ clauses of size $k$.
	Since a clause of size $r$ is falsified by $2^{n-1-r}$ assignments, and each assignment falsifies some clause, we have $p 2^{n-k} + (m-p-q) 2^{n-1-k} \geq 2^{n-1}$, and solving for $m$ completes the proof.
\end{proof}

\begin{lemma}
	\label{lemma:ks-mu-var-bound}
	A minimally unsatisfiable $(k, s)$-formula with $m$ clauses has between $\lceil\frac{k \cdot m}{s}\rceil$ and $m-1$ variables. Similarly, a minimally unsatisfiable $(k,p,q)$-formula with $m$ clauses has between $\lceil\frac{k \cdot m}{p+q}\rceil$ and $m-1$ variables.
\end{lemma}
\begin{proof}
	With $n$ variables of degree $\leq s$ there are $mk \leq ns$ literal occurrences.
	For the upper bound, recall that minimally unsatisfiable formulas have positive deficiency. 
\end{proof}

It is known that an unsatisfiable (3,4)-formula with 16 clauses  and an unsatisfiable (3,2,2)-formula with 20 clauses exist
\cite{BermanKarpinskiScott03c}. We give the formulas as  $E_{3,4}$ and $M_{3,2,2}$ in the appendix. The experimental results combined with this knowledge yield \Cref{table:prelim-findings}. One can find a smallest $(3,2,3)$, $(3,2,4)$ and $(3,3,3)$-formula in the appendix as $M_{3,2,3}$, $M_{3,2,4}$, and $M_{3,3,3}$. $M_{3,2,3}$ also serves as a smallest $(3,5)$-formula, and $M_{3,3,3}$ as a smallest $(3,6)$ formula.

\begin{table}[tbp]
  \caption{ Preliminary results based only on the method of this
    section, for $\mu(3,s)$ (left) and $\mu(3,p,q)$ (right). Since the
    right table is symmetric with respect to the main diagonal, we
    only give the upper triangle due to the assumption $q \geq p$. All values for $s \geq 6$ and $q \geq p \geq 3$ are $8$: less is not possible by Lemma~\ref{lemma:model-bound-asym}.
    Lemma~\ref{lemma:model-bound-asym} also rules out a
    $(3,2,q)$-formula with $9$ clauses and $q \geq 5$.}
  \label{table:prelim-findings}
\av{\medskip}
  \centering
  \setlength\tabcolsep{3mm}
  \av{\qquad}\begin{tabular}{@{}ccc@{}}
    \toprule
    $s=4$ & $s=5$ & $s \geq 6$ \\
    \midrule
    $[14,16]$ & 11 & 8\\
    \bottomrule\\
  \phantom{1,2,3} &&\\
  \end{tabular}\hfill%
  \begin{tabular}{@{}cccccc@{}}
    \toprule
    &$q=1$ & $q=2$ & $q=3$ & $q=4$ & $q\geq$5 \\
    \midrule
    $p=$ 1 & $\infty$ & $\infty$ & $[14,\infty]$& $[11,\infty]$ & $[8,\infty]$\\
    $p=2$ &-& $[8,20]$  & 11 & 10 & 10\\
    $p\geq 3$ & - & - & 8& 8 & 8\\
    \bottomrule
  \end{tabular}\av{\qquad}
\end{table}

A closer look at the time spent on deciding the existence of an unsatisfiable $(3,4)$-formula with different $n$ and $m$ shown in Table \ref{table:3-4-time} reveals that this basic method reaches its limit with formulas of about $13$ clauses. Thus, further considerations are called for if we want to determine the precise value for some of the entries in the tables. 

\newcommand{\unslvd}{to}
\newcommand{\unslvdw}{\multicolumn{2}{c}{\unslvd}}

\newcommand{\sep}{\hspace{3pt}}
\begin{table}[htb]
  \centering

    \caption{
		Time spent deciding the existence of an unsatisfiable $(3,4)$-formula with $n$ variables and~$m$ clauses (without/with blocked-clause detection encoded as $F_6$).
		Unsolved queries are marked by \unslvd, blank areas are out of bounds determined by Lemma~\ref{lemma:ks-mu-var-bound}.
		All terminated queries were unsatisfiable except the one marked
                in blue with $n=12$ and $m=16$.
	}
    \label{table:3-4-time}
\av{\medskip}

  \setlength\tabcolsep{2pt}
  \begin{tabular}{@{}lrr@{\hspace{0.7em}}rr@{\hspace{0.7em}}rr@{\hspace{0.7em}}rr@{\hspace{0.7em}}rr@{}}
          \toprule
          $n$   & \multicolumn{2}{c}{$m=8$}     & \multicolumn{2}{c}{$m=9$}     & \multicolumn{2}{c}{$m=10$}     & \multicolumn{2}{c}{$m=11$} & \multicolumn{2}{c}{$m=12$}\\
          \midrule
          6     & 0.4s & 0.6s &      &      &      &       &     &    &     &      \\
          7     & 0.7s & 1.2s & 1.9s & 2.0s &      &       &     &    &     &      \\
          8     &      &      & 4.9s & 5.9s & 7.3s & 10.4s &     &    &     &      \\
          9     &      &      &      &      &  42s & 1m    &  1m & 3m &  5m &  4m  \\
          10    &      &      &      &      &      &       & 12m & 7m & 43m &  17m \\
          11    &      &      &      &      &      &       &     &    &  3h & 1.6h \\
          \bottomrule
    \end{tabular}\hfill%
	\begin{tabular}{@{}l@{\hspace{0.7em}}rr@{\hspace{0.7em}}rr@{\hspace{0.7em}}rr@{\hspace{0.7em}}rr@{\hspace{0.7em}}@{}}
          \toprule
          $n$ & \multicolumn{2}{c}{$m=13$} & \multicolumn{2}{c}{$m=14$} & \multicolumn{2}{c}{$m=15$} & \multicolumn{2}{c}{$m=16$} \\
          \midrule
10  &  2h & 1h  &     &     &     &     &          &                       \\
11  & 12h & 6h  & 28h & 27h &     &     &          &                       \\
12  &  5d & 42h & \unslvdw  & \unslvdw  & \unslvd  & \textcolor{azure}{5d} \\
13  &     &     & \unslvdw  & \unslvdw  &       \unslvdw                   \\
14  &     &     &     &     & \unslvdw  &       \unslvdw                   \\
15  &     &     &     &     &     &     &       \unslvdw                   \\
    \bottomrule
    \end{tabular}
\end{table}

\vspace{-5mm}
\subsection{A dichotomy theorem for \texorpdfstring{$(3,1,q)$}{(3,1,q)}-SAT}

The following extends a result by Kratochv\'{\i}l, et
al.~\cite[Lemma 2.2]{KratochST93}.

\newcommand{\cvstar}{\textnormal{$\star$}}
\begin{lemma}
\cv{[\cvstar]}
\label{lem:unsat-np}
  Let $k\geq 3$ and $p,q\geq 1$ such that $p+q\geq 3$. If there exists
  an unsatisfiable $(k,p,q)$\hy formula, then $(k,p,q)$-\normalfont{SAT} is
  NP-hard.
\end{lemma}

\ifdefined\isarxiv
\begin{proof}
  The proof mimics that of~\cite{KratochST93}.
  We give a reduction from $k$-SAT, i.e., the satisfiability problem
  for $k$-CNF formulas.  Assume, w.l.o.g., that $q\geq 2$ and let $F$
  be an instance of $k$-SAT with $m$ clauses and $n$ variables.
  Assume there is a variable $x$ that appears in more than $p$ clauses
  positively or in more than~$q$ clauses negatively (if such $x$
  doesn't exist, then $F$ is an instance of $(k,p,q)$-SAT and we are
  done).

  Let $C_1,\dots,C_r \in F$ be the clauses that contain $x$ and
  $D_1,\dots,D_s \in F$ the clauses that contain~$\ol{x}$.  We take
  fresh variables $y_1,\dots,y_s$ and $z_1,\dots,z_r$ and replace $x$
  in $C_i$ by $\ol{y_i}$, $1\leq i \leq r$, and replace $\ol{x}$ in
  $D_i$ by $\ol{z_i}$, $1\leq i \leq s$. Let $F'$ be the new formula
  obtained from $F$ by this replacement.

  To make $F$ and $F'$ equisatisfiable, we need to ensure that all
  $y_i$ variables are logically equivalent and all the $z_i$ variables
  are logically equivalent with the negation of the $y_i$ variables.
  We achieve this by adding a circle of implications in terms of
  $2$-clauses. More precisely, we add the clauses
  $\{\ol{y_i},y_{i+1}\}$ for all $1\leq i < r$, the clause
  $\{\ol{y_r}, \ol{z_1}\}$, the clauses $\{z_i,\ol{z_{i+1}}\}$ for all
  $1\leq i < s$, and the clause $\{z_s, y_1\}$.  Let $F''$
  denote the formula we obtain by adding these 2-clauses to $F'$. We
  observe that $F$ and $F''$ are indeed equisatisfiable and all the
  variables $y_1,\dots,y_r$ and $z_1,\dots z_s$ appear exactly once
  positively and twice negatively, i.e., they satisfy the occurrence
  condition imposed by $p$ and~$q$. Similarly, we can
  eliminate all further variables that appear more than $p$ times
  positively or more than $q$ times negatively, eventually arriving at
  an equisatisfiable $(\leq k, p,q)$-formula $F^*$ which contains $m$
  $k$-clauses and at most $km$ 2-clauses (each $2$-clause can be traced
  back to a unique literal of $F$, and there are $km$ such literals).

  It remains to increase the size of the $2$-clauses to $k$.  We know
  that there exists an unsatisfiable $(k,p,q)$\hy formula, and hence
  there exists a minimally unsatisfiable $(k,p,q)$\hy formula. Call
  this formula $H$, and let $C$ be some fixed clause of $H$.
  We observe that $H' := H\setminus \{C\}$ is satisfiable and,
  for all $\ell \in C$, each satisfying truth assignment of $H'$ sets $\ell$ to
  false.  For each 2-clause $D$ of $F^*$ we add a copy of $H'$ to $F^*$,
  and add (copies of) $k-2$ literals from $C$ to $D$. Let $F^{**}$ denote the formula
  obtained this way from $F^*$. Again, $F^{**}$ and $F$ are
  equisatisfiable, and $F^{**}$ is now a $(k,p,q)$-formula. As there
  were at most $km$ 2-clauses in $F^*$, we have added at most $km(|H|-1)$
  additional clauses with the variable-disjoint copies of $H'$.
  
  Thus, we have reduced the given formula $F$ in polynomial time to an
  equisatisfiable $(k,p,q)$-formula $F^{**}$. Since $k$-SAT is
  NP-hard, it follows that $(k,p,q)$-SAT is NP-hard. 
\end{proof}
\fi
 
\begin{theorem}[Dichotomy]\label{the:dich}
  For any $p,q\geq 1$, if   $p+q < 4$ then $(3,p,q)$-\normalfont{SAT}
  is solvable in polynomial time, otherwise
  $(3,p,q)$-\normalfont{SAT}
  is NP-hard.
\end{theorem}
\begin{proof}
	$(3,p,q)$-\normalfont{SAT} is a special case of $(3,p+q)$-\normalfont{SAT}, and $(3,s)$-SAT is in P for
  $s\leq 3$~\cite{Tovey84}.
  Let $p+q \geq 4$, w.l.o.g., $p\leq q$.
  If $p=1$, then $q\geq 3$, if $p \geq 2$, then $q \geq 2$, and unsatisfiable $(3,1,3)$\hy formulas and $(3,2,2)$\hy formulas exist (which are also $(3,p,q)$\hy formulas for larger $p,q$, see Table~\ref{tab:3-s}).
  NP-hardness follows from \Cref{lem:unsat-np}.
\end{proof}
The NP-hardness part of Theorem~\ref{the:dich} holds
 even for monotone SAT,
where each clause is required to contain only positive or only
negative literals~\cite{DarmannDoecker21},
with the exception of  monotone $(3,1,3)$-SAT, for which van Santvliet and de
Haan~\cite{vansantvliet2023instances} have recently shown that all
instances are satisfiable, and monotone $(3,1,4)$-SAT, which is still open.
Our proof is uniform in the sense that all hardness results rely on~\Cref{lem:unsat-np}.

\section{Compound methods}
\label{sec:advanced}
The core method for finding small unsatisfiable $(k,s)$-CNF and
$(k,p,q)$-CNF formulas is our \emph{SMS encoding} that we have
introduced in \Cref{sec:encoding}. This method is quite powerful and
lets us produce the smallest formulas and exact lower bounds for formula
size. However, since the search space grows very quickly, this method
reaches its limits with formulas of about 10--15 clauses, depending on
the imposed side constraints. In this section, we show how we improve the results from Table~\ref{table:prelim-findings} by combining computational search with theoretical analysis, and with several techniques to decrease the size of the search space. The methods for obtaining upper bounds are \emph{disjunctive splitting} and \emph{combining enforcers}. The methods for obtaining lower bounds are \emph{reductions} and \emph{hard-coding part of the adjacency matrix}. 

Generating $(k, s)$-formulas directly is often prohibitively expensive.
In such cases we also consider $(\leq k,
s)$-formulas that have a few clauses of width smaller than
$k$.
Central to the various techniques we employ in this section is the concept of a stairway.  A \emph{stairway}, first introduced by Hoory and
Szeider~\cite{HoorySzeider05}, is an abstraction of a CNF formula that
focuses only on the clauses that are smaller than a given~$k$. More
specifically, a stairway $\sigma=(a_1,\dots,a_r)$ is a finite
non-increasing sequence of positive integers. For a fixed integer $k$,
a stairway $\sigma=(a_1,\dots,a_r)$ represents the set of all CNF formulas $F=\{C_1,\dots,C_m\}$ where $a_i=k-\Card{C_i}$ for $1\leq i \leq r$, and $\Card{C_i}=k$ for $r+1 \leq i \leq m$. Define $\mkssc{k}{s}{\sigma}$ to be the number of clauses of a smallest unsatisfiable $(\leq k, s)$-formula with stairway $\sigma$, and $\mkstsc{k}{p}{q}{\sigma}$ to be that of the smallest unsatisfiable $(\leq k,p,q)$-formula with stairway $\sigma$. 
 
It is straightforward to adapt $\ksnm{k}{s}{n}{m}$ to encode, instead of a $(k,s)$-formula, a $(\leq k,s)$-formula with a given stairway. Instead of requiring every clause to contain exactly $k$ literals, we require every clause to contain at most $k$ literals, and that a certain number of clauses contain less than $r$ literals, for some $1<r\leq k$. For each $1<r\leq k$, $j\in [n]$, let $N_r$ be the number of clauses that contains strictly less than $r$ literals. We replace $F_5$ with $F'_5$ below. 
\[F'_5:=\bigwedge_{j\in[n]}\lnot\tau_{j,k+1}\land\bigwedge_{1<r\leq k}\sum_{j\in[n]}\lnot\tau_{j,r}=N_r.\]

\subsection{Improved bounds}
\label{sec:theory}
The following theorems refine the bounds on $\mks{k}{s}$ and $\mkst{k}{p}{q}$ using stairways.

\begin{theorem}
\label{k-p-q-theorem}
$\mu(k,p,q)\geq \min \big ( \mu(k,p,q,1^p) + q, \mu(k,p,q-1) \big )$,
where $1^p$ is the stairway of length~$p$ with each entry being a $1$.
\end{theorem}
\begin{proof}
An unsatisfiable $(k,p,q)$-formula where at least one literal occurs $q$ times gives an unsatisfiable formula with $p$ ($k-1$)-clauses and shorter in length by $q$ if we set the literal that occurs $q$ times to true. The latter has minimal size $\mu(k,p,q,1^p)$. Taking this into account, we know that $\mu(k,p,q)<\mu(k,p,q-1)$ is only possible if $\mu(k,p,q,1^p) + q < \mu(k,p,q-1)$.
\end{proof}

Recall the concept of an enforcer from Section~\ref{sec:prelim}.
Enforcers can be used to provide upper bounds, as we will show in Theorem~\ref{enforcer-theorem}.
When $k=3$, a
$(3,s)$-enforcer (or a $(3,p,q)$-enforcer) gives rise, by appending the appropriate unit clause, to an
unsatisfiable $(\leq k,s)$-formula (or an unsatisfiable
$(\leq k,p,q)$-formula) with stairway~$(2)$, and thus by searching for formulas with this stairway we can generate enforcers.
In this way, we computed the size of
a smallest enforcer in the classes of $(3,s)$-formulas and
$(3,p,q)$-formulas, and list them in Table \ref{tab:3-enf}.
The minimality of the $(3,3,4)$\hy enforcer was also shown by Jurenka~\cite{Jurenka11} with
a theoretical argument.  The corresponding formulas can be found in
Appendix~\ref{appendix:enforcers}.
\begin{theorem}
\label{enforcer-theorem}
$\mu(k,s)\leq k\cdot(\mu(k,s,(k-1)) - 1) + 1$. Similarly, $\mu(k,p,q)\leq k\cdot(\mu(k,p,q,(k-1)) - 1) + 1$.
\end{theorem}
\begin{proof}
Let $E$ be a smallest $(k,s)$-enforcer, which by definition has size $\mu(k,s,(k-1)) - 1$ (the $(k,p,q)$ case is analogous).
We can obtain an unsatisfiable $(k,s)$-formula by taking $k$ variable-disjoint copies of $E$ and adding a $k$-clause containing the negated enforced literals. 
\end{proof}

\begin{table}[htb]
    \caption{Smallest size of an unsatisfiable  $(\leq k,1,q)$\hy
      formula for  stairway~(1). 
      }
      \label{tab:stairway_1}
      \av{\medskip}
    \centering
    \setlength\tabcolsep{3mm}
    \begin{tabular}{@{}cccccc@{}}
    \toprule
&$q=1$ & $q=2$ & $q=3$ & $q=4$ & $q \geq 5$ \\
\midrule
$p=1$ & $\infty$ & $\infty$ & 14 or 15 & 13 & 11\\
\bottomrule
    \end{tabular}
\end{table}
We compute the table for $\mkstsc{3}{1}{3}{(1)}$ and
$\mkstsc{3}{1}{3}{(2)}$ through a similar exhaustive search as explained in Section~\ref{sec:compute}, and show the results in Table \ref{tab:stairway_1}.
Combining these tables with Theorems~\ref{k-p-q-theorem} and \ref{enforcer-theorem}, we have the following improvement.

\begin{corollary}
\label{corollary:3-1-3-lower-bound}
$\mkst{3}{1}{3},\mkst{3}{1}{4}\geq 17$ and $\mkst{3}{1}{q}\geq 16$ for all $q\geq 5$.
\end{corollary}

\begin{corollary}
\label{corollary:3-1-3-upper-bound}
$\mkst{3}{1}{3} \leq 22, \mkst{3}{1}{4}\leq 19$ and $\mkst{3}{1}{q}\leq 16$ for all $q\geq 5$.
\end{corollary}

\subsection{Disjunctive splitting}
\label{sec:splitting}
We say a CNF formula $F$ is obtained by \emph{disjunctive splitting in $x$} from CNF formulas $F_1,F_2$, in symbols $F=F_1 \oplus F_2$, if $F$ can be partitioned into two nonempty sets $F_1',F_2'$ such that the variable $x$ occurs in $F_1'$
positively but not negatively, and appears in $F_2'$ negatively but
not positively, and $F_i$ is obtained from $F_i'$ with all occurrences
of $x,\ol{x}$ removed. Observe that if $F_1,F_2$ are unsatisfiable,
then also $F= F_1' \cup F_2'$ is unsatisfiable. Hence, when
constructing an unsatisfiable $(k,s)$\hy CNF or $(k,p,q)$-CNF formula,
we can first try to construct $(\leq k,s)$-formulas or $(\leq k,p,q)$-formulas $F_1,F_2$ and then combine them to obtain~$F$.

If $F$ is obtained by disjunctive splitting in $x$ from $F_1,F_2$, and
$x$ is added positively to $p$ clauses in $F_1$ and negatively to $q$
clauses in $F_2$, we write $F = F_1 \oplus_{p,q} F_2$. Disjunctive
splitting can be recursively applied to $F_1$ and $F_2$.  This allows
us to construct a formula from \emph{axioms} which are CNF formulas that we do not further split.  For example, the $(2,4)$-CNF formula
$\{\{x,y\}$, $\{\ol{x},y\}$, $\{\ol{y},z\}$, $\{\ol{y},\ol{z}\}\}$ can be
constructed from the axiom $\{\nil\}$. 

We can describe the
construction by an \emph{$\oplus$-derivation}, an algebraic expression
$ (\{\nil\} \oplus_{1,1} \{\nil\}) \oplus_{2,2} (\{\nil\} \oplus_{1,1}
\{\nil\})$. In fact, $\MU(1)$ is exactly the class of all formulas
that can be constructed by disjunctive splitting from the axiom
$\{\nil\}$.

This idea was utilized by Hoory and Szeider~\cite{HoorySzeider05}, who
proposed an algorithm that decides for given $k,s$ whether
$\CNF{(k,s)}\cap \MU(1) \neq \emptyset$. This allows us to compute an
upper bound on the threshold function $f(k)$.  
Hoory and Szeider define $\oplus$-derivations to operate on stairways instead of formulas. 
For $k=3$, the above $\oplus$-derivation would now read
$ ((3) \oplus_{1,1} (3)) \oplus_{2,2} ((3) \oplus_{1,1} (3))$ and
produce the stairway $(1,1,1,1)$.  By means of a saturation algorithm, Hoory and
Szeider could determine the upper bounds 
$f(3)\leq 3$, $f(4)\leq 4$, $f(5)\leq 7$, $f(6)\leq 11$,
$f(7)\leq 17$, $f(8)\leq 29$, and $f(9)\leq 51$
on the threshold function $f(k)$ (i.e.,
all $(k,f(k))$-formulas are satisfiable but $(k,f(k)+1)$-SAT
is NP-complete), which are still the
best known upper bounds.

In this paper, we generalize Hoory and Szeider's method in the following ways:
\av{
\begin{enumerate}
	\item we consider $\oplus$-derivations with more axioms: any unsatisfiable $(\leq k,s)$-formula can serve as an axiom;
	\item we modify the saturation algorithm so that it gives the size of the smallest unsatisfiable $(k,s)$-formula derivable with respect to the sizes of the formulas that serve as axioms;
	\item we adapt the algorithm also for searching for unsatisfiable $(k,p,q)$-formulas.
        \end{enumerate}}
      \cv{(i)~we consider $\oplus$-derivations with more axioms: any unsatisfiable $(\leq k,s)$-formula can serve as an axiom;
	(ii)~we modify the saturation algorithm so that it gives the size of the smallest unsatisfiable $(k,s)$-formula derivable with respect to the sizes of the formulas that serve as axioms;
	(iii)~we adapt the algorithm also for searching for unsatisfiable $(k,p,q)$-formulas.}
		We do this in the hope that, by finding suitable axioms with the SMS encoding, we can incorporate them into the $\oplus$-derivations to obtain smaller unsatisfiable $(k,s)$ or $(k,p,q)$-formulas.

Table \ref{table:smallest-MU-1} shows the size of smallest unsatisfiable $(3,s)$ and $(3,p,q)$-formulas generated by disjunctive splitting with $\{\nil\}$ being the only axiom. The corresponding $\oplus$-derivations can be found in the appendices. We determined $\mu(3,s,1^r)$ for all $1\leq r<s$ and $\mu(3,p,q,1^r)$ for all $1\leq r<q$, but this did not yield any new upper bounds. We were more successful with the method of disjunctive splitting in the case of $k=4$, which we discuss in Section \ref{sec:smaller_4_5}.

\begin{table}[htb]
      \setlength\tabcolsep{1.9mm}
\caption{Size of the smallest  $(3,s)$-formulas and $(3,p,q)$-formulas in $\MU(1)$.} 
\label{table:smallest-MU-1}
\av{\medskip}
\centering
  \begin{tabular}{@{}cccccc@{}}
\toprule
$s\leq 3$ & $s=4$ & $s=5$ & $s=6$ & $s=7$ & $s\geq 8$\\
\midrule
$\infty$ & 16 & 12 & 10 & 9 & 8\\ 
  \bottomrule\\
  \phantom{123} &&&&&\\
  \phantom{123} &&&&&\\
\end{tabular}\hfill%
\begin{tabular}{@{}cccccc@{}}
    \toprule
&$q=1$ & $q=2$ & $q=3$ & $q=4$ & $q\geq 5$ \\
\midrule
$p= 1$ & $\infty$ & $\infty$ & 22 & 19 & 16\\
$p=2$ &-& $\infty$ & 12 & 10 & 10\\
$p=3$ & - & - & 10 & 9 & 9 \\
$p\geq 4$ & - & - & - & 8 & 8 \\
\bottomrule
    \end{tabular}
\end{table}

\subsection{Hard-coding part of the matrix}
After we narrowed down the search scope with tighter bounds, this and the next part deal with deciding the missing values in Table~\ref{tab:3-s} and the techniques involved. These techniques allow us to determine the existence of unsatisfiable $(3, 4)$-formulas (or $(3,2,2)$,$(3,1,q)$-formulas) whose size is too large to be exhaustively searched by SMS directly. 

In this part, we determine the value of $\mu(3,1,3)$ and $\mu(3,1,3)$. The technique here is hard-coding the part of the matrix that corresponds to the occurrences of some variables/literals. The motivation is that with less undecided values in the matrix to solve, SMS terminates more quickly. Suppose we want to fix both positive and negative occurrences of $\fixvar$ variables $x_1,x_2,\dots,x_{\fixvar}$ and only the positive occurrences of $\fixlit$ other variables $x_{\fixvar+1},x_{\fixvar+2},\dots,x_{\fixvar+\fixlit}$. To do this in a way that is compatible with the minimality check of SMS, we need to adjust the following two things. First, when generating the encoding, we stipulate that the first rows in the matrix correspond to $x_1,\overline{x}_1,x_2,\overline{x}_2,\dots,x_{\fixvar},\overline{x}_{\fixvar},x_{\fixvar+1},x_{\fixvar+2},...x_{\fixvar+\fixlit}$ and adjust pos and neg to the following.
\[\pos{i}:=\begin{cases}2i-1&\text{if $i\leq
    \fixvar$},\\i+\fixvar&\text{otherwise;}\end{cases}\qquad
  \nega{i}:=\begin{cases}2i& \text{if $i\leq \fixvar$,}\\2n+\fixvar-i+1&\text{otherwise.}\end{cases}\] Second, we start SMS with the minimality check restricted to the refined partition \[\{\{1\},\{2\},\dots,\{2\fixvar+\fixlit\},[2\fixvar+\fixlit+1,2n],[2n+1,2n+m]\} \] of the original $\{[1,2n],[2n+1,2n+m]\}$. Given the specific assumption on the variables and literals fixed, the minimality condition will determine the values of the first $2\fixvar+\fixlit$ rows in the matrix. An example to this will be given shortly after.
  
\begin{figure}
    \centering
    \small
\begin{tblr}{@{}c|c|[dashed]ccc|[dashed]ccc|[dashed]ccc|[dashed]c@{}}
$\matr{i}{2n+j}$&$1$\dots&$m-9$&&&$m-6$&&&$m-3$&&&$m$\\
\hline
$1\ (x_1)$&&&&&&&&&&&1\\
$2\ (\n{x_1})$&&&&&&&&1&1&1&\\
\hline[dashed]
$3\ (x_2)$&&&&&&&&&&&1\\
$4\ (\n{x_2})$&&&&&1&1&1&&&&\\
\hline[dashed]
$5\ (x_3)$&&&&&&&&&&&1\\
$6\ (\n{x_3})$&&1&1&1&&&&&&&\\
\hline[dashed]
$7\ (a_1)$&&&&&&&&&&1&\\
$8\ (b_1)$&&&&&&&&&&1&\\
\hline[dashed]
$9\ (a_2)$&&&&&&&1&&&&\\
$10\ (b_2)$&&&&&&&1&&&&\\
\hline[dashed]
$11\ (a_3)$&&&&1&&&&&&&\\
$12\ (b_3)$&&&&1&&&&&&&\\
\end{tblr}
	\caption{The first rows of the matrix determined as a result of our choice for $\pos{i}$ and $\nega{i}$. The omitted values are all $0$.}
    \label{fig:fixed-matrix}
\end{figure}
To determine the value of $\mu(3,1,3)$ and $\mu(3,1,4)$, we prove the following lemma that argues about the presence of certain ``partially known'' $\leq k$-CNFs, so we can fix them. As explained in Section~\ref{par:prelims-cnf}, we write $\leq k$-CNFs in matrix form.  To argue about
``partially known'' $\leq k$-CNFs, we extend the
matrix notation to what is essentially a first-order language of
$\leq k$-CNFs.  We use lowercase letters $x,y,z,a,b,c, \dots$ to
denote symbols to be interpreted by propositional literals, positive
or negative.  We use $\n{\phantom{x}}$ to denote negation: if $x$ is
interpreted as some literal, then $\n{x}$ must be interpreted as its
negation.  Two different literal symbols may be interpreted by
different literals, or they may be interpreted by the same literal, or
even by the two literals of the same variable.  When a position in the
matrix is left blank, we leave the corresponding literal
unconstrained.  We then say a formula $F$ is \emph{of the form} $M$ if
the symbols in the matrix $M$ can be interpreted by the literals of
$F$ to yield a matrix of $F$.

We say a clause $C\in F$ is \emph{singular} if $\ldeg{x}=1$ for all $x\in C$.

\begin{lemma}
\cv{[\cvstar]}
\label{lemma:fixing-3-1-3}
Let $q\geq 3$ and let $F$ be a smallest unsatisfiable $(3,1,q)$-formula.
\begin{enumerate}
\item There is a singular clause in $F$.
\item If $\begin{psmallmatrix}
    x_1 \\
    x_2 \\
    x_3 \\
\end{psmallmatrix}\in F$ is singular, then $F[x_i]\cap F[x_j]=\left\{\begin{psmallmatrix}
    x_1 \\
    x_2 \\
    x_3 \\
\end{psmallmatrix}\right\}$ for any $i\not=j\in\{1,2,3\}$, and $\ldeg{\n{x_1}}=\ldeg{\n{x_2}}=\ldeg{\n{x_3}}\geq 3$.
\item Let $\begin{psmallmatrix}
    x_1 \\
    x_2 \\
    x_3 \\
\end{psmallmatrix}\in F$ be a singular clause and let $a$ be a literal such that $\ldeg{a} = 1$ and $\var(a)\not=\var(x_i)$ for all $i\in\{1,2,3\}$. If $a\in\bigcup F[x_i]$, then $\n{a}\not\in\bigcup F[x_j]$ for any $i\not=j\in\{1,2,3\}$. 
\item If there is a unique singular $\begin{psmallmatrix}
    x_1 \\
    x_2 \\
    x_3 \\
\end{psmallmatrix}\in F$, then $\bigcup_{i=1,2,3}F[x_i]$ is of the form \\$\begin{psmallmatrix}
    x_1 & \overline{x}_1 & \overline{x}_1 & \overline{x}_1 & \overline{x}_2 & \overline{x}_2 & \overline{x}_2 & \overline{x}_3 & \overline{x}_3 & \overline{x}_3 \\
    x_2 & a_1 & & & a_2 & & & a_3\\
    x_3 & b_1 & & & b_2 & & & b_3\\
\end{psmallmatrix}$ where  $\ldeg{a_i}=\ldeg{b_i}=1$ for all $i\in\{1,2,3\}$.
\end{enumerate}
\end{lemma}

\av{
\begin{proof}
For the first part, note that any unsatisfiable $(k,1,q)$-formula must have a clause consisting of literals that occur only once, otherwise an assignment $\alpha$ that sets all literals that occur more than once to true is an satisfying assignment (variables of type $(1, 1)$ can be assigned arbitrarily).

For the second part, let $F$ be a smallest unsatisfiable $(3,1,q)$-formula and let $\begin{psmallmatrix}
    x_1 \\
    x_2 \\
    x_3 \\
\end{psmallmatrix}\in F$ such that $\ldeg{x_1}=\ldeg{x_2}=\ldeg{x_3}=1$. Then no $\overline{x}_i,\overline{x}_j$ for any $i\not= j\in\{1,2,3\}$ can occur together in a clause, otherwise that clause is blocked on $x_i$ and therefore can be removed to obtain a smaller unsatisfiable $(3,1,q)$-formula.
Also, $\overline{x}_i$ has to occur at least 3 times, otherwise we can obtain a smaller unsatisfiable $(3,1,q)$-formula by eliminating $x_i$. So far, we have established part~2.

For all $i\in\{1,2,3\}$, let $F_i$ be the unsatisfiable formula obtained from $F$ by setting $x_i$ and $\overline{x}_j$ for all $j\not=i\in\{1,2,3\}$ to true.
Note that $F$ is unsatisfiable if and only if all $F_1$, $F_2$ and $F_3$ are unsatisfiable.

To see the third part, suppose $a\in\bigcup F[x_i]$ and $\n{a}\in\bigcup F[x_j]$ for some $i\not=j\in\{1,2,3\}$. Since $\ldeg{a}=1$, it follows that $a$ does not occur in $F_j$.
This means that $F_j$ is still unsatisfiable if we remove the clauses in $F[x_j]$ that contain $\n{a}$ from it, which in turn implies that $F$ is still unsatisfiable if we remove from $F$ the clauses in $F[x_j]$ that contain $\n{a}$.
This contradicts our assumption that $F$ is a smallest unsatisfiable $(3,1,q)$-formula.

For the last part, suppose there is a unique singular $\begin{psmallmatrix}
    x_1 \\
    x_2 \\
    x_3 \\
\end{psmallmatrix}\in F$. Let $F':=F/(F[x_1]\cup F[x_2]\cup F[x_3])$. Note that the assignment $\alpha$ mentioned in part 1 satisfies all clauses in $F'$. This means that for each $i\in\{1,2,3\}$, in order for $F_i$ to be unsatisfiable, there has to be $\begin{psmallmatrix}
    \n{x_i} \\
    a_i \\
    b_i \\
\end{psmallmatrix}\in F[x_i]$ such that $\ldeg{a_i}=\ldeg{b_i}=1$.
\end{proof}
}

When searching for a smallest unsatisfiable $(3,1,3)$-formula, we distinguish the following two cases, depending on whether there is a unique singular clause. If there is, then we hard-code the partial formula $\begin{psmallmatrix}
    x_1 & \overline{x}_1 & \overline{x}_1 & \overline{x}_1 & \overline{x}_2 & \overline{x}_2 & \overline{x}_2 & \overline{x}_3 & \overline{x}_3 & \overline{x}_3 \\
    x_2 & a_1 & & & a_2  & & & a_3 & & \\
    x_3 & b_1 &  &  & b_2& & & b_3 &  & \\
\end{psmallmatrix}$ in the matrix by setting the first $12$ rows to represent $x_1,\overline{x}_1,x_2,\overline{x}_2,x_{3},\overline{x}_{3},a_1,b_1,a_2,b_2,a_3$ and $b_3$,
starting with the ordered partition $\{\{1\},\{2\},\dots,$ $
\{12\},[13,2n],[2n+1,2n+m]\}$, and hard-coding the first $12$ rows of the matrix thereby determined. The fixed rows are shown in Figure~\ref{fig:fixed-matrix} as an example. Otherwise, there are more than one singular clause and we hard-code the partial formula $\begin{psmallmatrix}
    x_1 & \overline{x}_1 & \overline{x}_1 & \overline{x}_1 & \overline{x}_2 & \overline{x}_2 & \overline{x}_2 & \overline{x}_3 & \overline{x}_3 & \overline{x}_3 & y_1\\
    x_2 & & & & & & &  & & & y_2\\
    x_3 & & & & & & & & & & y_3\\
\end{psmallmatrix}$. When searching for a smallest unsatisfiable $(3,1,4)$-formula, we follow the same rationale but split both of these cases into 4 cases depending on the size of $n_5^*=\{x_i\,|\,\vdeg{x_i}=5,i\in\{1,2,3\}\}$.
\begin{table}
    \caption{Time spent in search for unsatisfiable
      $(3,1,3)$-formulas for different numbers of variables and
      clauses after hard-coding the first rows as described. All queries were unsatisfiable except the one marked in blue with $m = 22$ and $n = 21$. For each $m$, the lower bound for the
      choice of $n$ is by Lemma \ref{lemma:ks-mu-var-bound}, and the
      upper bound is $m-2$ since we know the smallest size of a
      $(3,1,3)$-formula in
      $\MU(1)$. There are two values for each pair of $m$ and $n$. The values on the left indicate the time spent on the case where we assume there is a unique singular clause. The values on the right indicate that of the case where we assume there are at least two singular clauses. The cases superscribed with $\star$ and $\blacktriangle$ did not terminate within a timeout of 6 days and were further split into sub-cases and time shown is the sum of time spent on each of the sub-cases. The cases superscribed with $\star$ are split into 1258 cases in terms of where $\n{a_1},\n{b_1},\n{a_2}$ and $\n{a_2}$ occur, modulo symmetries. The ones superscribed with $\blacktriangle$ are split into 136 cases in terms of where $\n{y_1}, \n{y_2}$ and $\n{y_3}$ occur, modulo symmetries. 
      It is worth noting that both sets of case distinctions are generated automatically without symmetry by reformulating them as graph problems and giving the corresponding encoding to SMS. An unsatisfiable formula is found in this case, assuming there exists a unique singular clause. Upon inspection, this formula is revealed to be composed of 3 enforcers in the spirit of Theorem~\ref{enforcer-theorem}.} 
    \label{tab:time-3-1-3}
    \centering
\av{\medskip}
    \setlength\tabcolsep{2.3mm}

	\begin{tabular}{@{}rr@{\hspace{0.4em}}rr@{\hspace{0.4em}}rr@{\hspace{0.4em}}rr@{\hspace{0.4em}}rr@{\hspace{0.4em}}rr@{\hspace{0.4em}}r@{}}
    \toprule
    & \multicolumn{2}{c}{$m=17$}& \multicolumn{2}{c}{$m=18$}& \multicolumn{2}{c}{$m=19$}& \multicolumn{2}{c}{$m=20$}& \multicolumn{2}{c}{$m=21$}& \multicolumn{2}{c}{$m=22$}\\
    \midrule
		 $n=13$&5.3s & 7.4m&  &  &  &  &  &  &  &  &  &  \\
	$n=14$&16s & 17m& 38s & 1.8h&  &  &  &  &  &  &  &  \\      
	$n=15$&19m & 1.1h&11m & 28.7h &12m & 43.8h &53m & 13.7h${}^{\star}$ &  &  &  &  \\
	$n=16$&  &  &7.2m & 20h &7.4h & 5.8h${}^{\star}$ &1.7d & 1.2d${}^{\star}$ & 1.7d & 9.9d${}^{\star}$ &  & \\
    $n=17$&  &  &  &  &3.8d & 9.0h${}^{\star}$ & 15.9h$^{\blacktriangle}$ & 1.8d${}^{\star}$&2.1d$^{\blacktriangle}$ & 14.3d${}^{\star}$&  &  \\
    $n=18$&  &  &  &  &  &  &1.2d$^{\blacktriangle}$ & 2.5d${}^{\star}$&3.6d$^{\blacktriangle}$ & 34.3d${}^{\star}$&  &  \\
    $n=19$&  &  &  &  &  &  &  &  &11.6d$^{\blacktriangle}$ & 56.9d${}^{\star}$&  &  \\
    $n=21$&  &  &  &  &  &  &  &  &  &  &\textcolor{azure}{29.2m} & - \\
    \bottomrule
    \end{tabular}
\end{table}

\begin{table}[!ht]
    \caption{Time spent in search for unsatisfiable
      $(3,1,4)$-formula for different numbers of variables and
      clauses after hard-coding the first rows as described. The timeout is 11 days. For each $m$ in the table, the lower bound for the choice of $n$ is by Lemma \ref{lemma:ks-mu-var-bound}, and the
      upper bound is $m-2$ since we know that the smallest size of a $(3,1,4)$-formula in $\MU(1)$. There are eight values for each pair of $m$ and $n$ divided into two rows of four values each. The row on the top shows the amount of time spent on the cases where we assume there is a unique singular clause. The row on the bottom shows the amount of time spent on the cases where we assume there are at least two singular clauses. The four values in every row each indicates the time spent on the case assuming that $n_5^*$ of out the three variables $x_1,x_2,x_3$ are of degree $5$. An unsatisfiable formula is found only for the case where $m = 19$, $n = 18$ and $\vdeg{x_1}=\vdeg{x_2}=\vdeg{x_3}=5$. Upon inspection, this formula is revealed to be composed of 3 enforcers in the spirit of Theorem~\ref{enforcer-theorem}.}
    \label{tab:time-3-1-4}
    \centering
\av{\medskip}
    \setlength\tabcolsep{2.3mm}
 \begin{tabular}{rp{0.5em}@{}r@{\hspace{0.4em}}r@{\hspace{0.4em}}r@{\hspace{0.4em}}rp{0.5em}@{}r@{\hspace{0.4em}}r@{\hspace{0.4em}}r@{\hspace{0.4em}}rp{0.5em}@{}r@{\hspace{0.4em}}r@{\hspace{0.4em}}r@{\hspace{0.4em}}r}
    \toprule
	&& \multicolumn{4}{c}{$m=17$} && \multicolumn{4}{c}{$m=18$} && \multicolumn{4}{c}{$m=19$} \\
    \cmidrule(r){3-6}\cmidrule(r){8-11}\cmidrule(r){13-16}
    $n_5^*$&& $0$& $1$ & $2$ & $3$ && $0$ & $1$ & $2$ & $3$ && $0$ & $1$ & $2$ & $3$\\
    \cmidrule(r){3-6}\cmidrule(r){8-11}\cmidrule(r){13-16}
    \multirow{2}{*}{$n=11$}& & 1.1s & 1.1s & 1.1s & 0.7s && - & - & 1.1s & 0.8s &&  &  &  &  \\
    && 19s & 15s & 12s & 13s && - & - & 13s & 14s &&  &  &  &  \\
    \cmidrule(r){3-6}\cmidrule(r){8-11}\cmidrule(r){13-16}
    \multirow{2}{*}{$n=12$}& & 1.8s & 2.0s & 2.0s & 2.0s && 2.2s & 2.5s & 2.4s & 2.7s &&  4.7s & 2.7s & 2.8s & 3.0s \\ 
    && 1.1m & 1.0m & 0.6m & 13s && 1.2m & 1.2m & 0.5m & 0.5m && 2.0m & 1.4m & 1.9m & 1.6m  \\
    \cmidrule(r){3-6}\cmidrule(r){8-11}\cmidrule(r){13-16}
    \multirow{2}{*}{$n=13$}& & 15s & 7.1s & 4.0s & 6.4s && 5.5s & 10s & 12s & 8.2s && 12s & 9.0s & 6.0s & 14s \\
    && 4.5m & 1.1m & 0.5m & 1.0m && 6m & 4.9m & 0.9m & 1.5m && 11m & 5.1m & 10m & 9.5m  \\
    \cmidrule(r){3-6}\cmidrule(r){8-11}\cmidrule(r){13-16}
    \multirow{2}{*}{$n=14$}& & 24s & 43s & 14s & 12s && 25s & 56s & 26s & 32s && 1.2m & 40s & 1.2m & 34s  \\  
    && 9.1m & 2.4m & 1.2m & 1.1m && 34m & 8.4m & 11m & 3.6m && 1.7h & 3.0h & 2.8h & 0.9h \\
    \cmidrule(r){3-6}\cmidrule(r){8-11}\cmidrule(r){13-16}
    \multirow{2}{*}{$n=15$}& & 1.0m & 1.7m & 17s & 21s && 17m & 2.7m & 49s & 1.8m && 2.3m & 1.0h & 1.3m & 9.7m  \\
    && 42m & 10m & 2.3m & 1.9m && 1.7h & 50m & 22m & 22m && 33.2h & 17.1h & 2.6h & 1.4h  \\
    \cmidrule(r){3-6}\cmidrule(r){8-11}\cmidrule(r){13-16}
    \multirow{2}{*}{$n=16$}& & &  &  & && 39m & 24m & 36m & 3.3m && 14.1h & 4.5h & 34m & 30m  \\ 
    && &  &  & && 19.7h & 47m & 24m & 40m && t.o. & 6.4d & 3.3h & 2.5h  \\
    \cmidrule(r){8-11}\cmidrule(r){13-16}
    \multirow{2}{*}{$n=17$}& & &  & & & & &  &  & && 7.2d & 12.0h & 4.8h & 2.6h  \\
    && &  &  & & &  &  & & & & t.o. & t.o. & 10.1h & 2.3h  \\
    \cmidrule(r){13-16}
    \multirow{2}{*}{$n=18$}& & &  & & & & & & &  & & t.o. & 5.6d & 6.5h & \textcolor{azure}{1.1h}  \\  
    && & & &  & & &  & & & & 2.4d & t.o. & 6.3d & 10.4h  \\
    \bottomrule
\end{tabular}
\end{table}

Combining the computational results in Tables~\ref{tab:time-3-1-3} and~\ref{tab:time-3-1-4} with the previous lemma, we have the following result.
\begin{theorem}
$\mu(3,1,3)\geq 22$ and $\mu(3,1,4)\geq 19$.
\end{theorem}

\subsection{Reduction}
In this section, we describe the final technique that allow us to determine the value of $\mu(3,4)$ and $\mu(3,2,2)$, and thus complete Table~\ref{tab:3-s}. The idea is to reduce a $(3,4)$-formula (or $(3,2,2)$-formula) to a smaller $(\leq 3,4)$-formula (or $(\leq 3,2,2)$-formula) that is equisatisfiable, so that the question of the existence of a certain formula is reduced to that of the existence of a certain, smaller formula. 
We then use SMS to determine the existence of such small formulas.

It is difficult to prove any useful properties about general unsatisfiable $(3,4)$-formulas (or $(3,2,2)$-formulas), but since we exhaustively search from smaller to bigger formulas, we can restrict our search to \emph{minimal} (in terms of the number of clauses) $(3,4)$-formulas (or $(3,2,2)$-formulas). We reduce such an $F$ to a smaller unsatisfiable $(\leq 3,4)$-formula  (or $(\leq 3,2,2)$-formula) by replacing all subsets of clauses from $F$ that fit into one of the two forms below with a single $2$-clause $\begin{psmallmatrix}
    c\\
    d\\
    \times\\
\end{psmallmatrix}$.
For each replacement operation, the symbols $c$, $d$ are instantiated separately, i.e., they could be instantiated differently each time.
Each replacement is tantamount to a sequence of variable eliminations, and thus is sound (preserves unsatisfiability).
\begin{enumerate}
\item $\begin{psmallmatrix}
    x & \n{x} & \n{x} & c\\
    a & a & a & \n{a} \\
    d & y & \n{y} & d\\
    \end{psmallmatrix}$ for some $\vdeg{x}=3$, $\vdeg{y} = 2$ and $d$.
	Eliminate: $\var(y), \var(x), \var(a)$.
\item $\begin{psmallmatrix}
    y & \n{y}\\
    c & c\\
    d & d\\
    \end{psmallmatrix}$ for some $\vdeg{y}=2$ and $\vdeg{c}=\vdeg{d}=4$.
	Eliminate: $\var(y)$.
\end{enumerate}

Given the number of clauses $m$ and the number of variables $n$ of the formula, Lemma~\ref{lem:at-most-one-deg-2} and Lemma~\ref{lemma:minimal-formula} help us narrow down possibilities for the number of subsets that fit in each of the forms.%
\footnote{
	Compared to the conference version of this paper, we have introduced \Cref{lem:at-most-one-deg-2} and \Cref{lemma:minimal-formula}.
	This allowed us to simplify the case distinction presented in \Cref{lemma:3-4-profiles} and \Cref{table:3-4-cases}.
}
\begin{lemma}
\label{lem:at-most-one-deg-2}
Let $F$ be an unsatisfiable $(3,4)$-formula (or $(3,2,2)$-formula). Then there is an unsatisfiable $(3,4)$-formula (or $(3,2,2)$-formula) $F'$ such that $|F'| \leq |F|$ and $F'$ has at most one variable of degree $2$.
\end{lemma}
\begin{proof}
If $F$ has at most one variable of degree $2$, then let $F':=F$ and we are done. Otherwise, let $\var(x)\not=\var(y)$ be two variables of degree $2$ in $F$. Without loss of generality, $F[x]\cup F[y]$ is of the form
$\begin{psmallmatrix}
    x & \n{x} &y & \n{y}\\
    \\
    \\
    \end{psmallmatrix}$, $\begin{psmallmatrix}
       x  & \n{x} & y \\
	\n{y} &           \\
    \\
    \end{psmallmatrix}$, or $\begin{psmallmatrix}
       x  & \n{x} \\
    \n{y} &    y  \\
    \\
\end{psmallmatrix}$. We can thus replace $y$ with $x$ and $\n{y}$ with $\n{x}$ in $F$, and delete tautological clauses if any arise, to get an unsatisfiable $(3,4)$-formula (or $(3,2,2)$-formula) with fewer variables of degree $2$. The new formula remains unsatisfiable, because identifying variables is semantically equivalent to restricting the set of assignments; it is easy to see that the constraints on clause size and literal occurrences are satisfied. Repeating this process as long as there is more than one variable of degree $2$ ultimately yields~$F'$.
\end{proof}

\begin{lemma}
\cv{[\cvstar]}
\label{lemma:minimal-formula}
    Let $F$ be an unsatisfiable $(3,4)$-formula (or $(3,2,2)$-formula) of minimal size. Let $x,y,w$ be literals such that $\var(x) \neq \var(w)$, 
    $\vdeg{y}=2$ and $\vdeg{x}=\vdeg{w}=3$. Then:
    \begin{enumerate}
    \item If a literal $a$ occurs in $F$, then $\overline{a}$ also occurs in $F$. \label{case:no-pure}
    \item If $\begin{psmallmatrix}
    a\\
    b\\
    \\
    \end{psmallmatrix}\in F$, then $\var(a)\not=\var(b)$.
    \label{case:diff-vari}
    \item $\occ{F}{y}$ is of the form $\begin{psmallmatrix}
    y & \n{y}\\
    a & a\\
    b & b\\
    \end{psmallmatrix}$. \label{case:var-deg-2}
    \item $\occ{F}{x}$ is of the form $\begin{psmallmatrix}
    x & \n{x} & \n{x}\\
    a & a & a\\
    b & c & d\\
    \end{psmallmatrix}$. 
	\label{case:var-deg-3}
    \item $\occ{F}{x}\cap\occ{F}{w}=\emptyset$.
    \label{case:disjoint-deg-3}
    \item Either $\occ{F}{x}\cap\occ{F}{y}=\emptyset$, or $\occ{F}{x}$ is of the form $\begin{psmallmatrix}
    x & \n{x} & \n{x}\\
    a & a & a\\
    z & y & \n{y}\\
    \end{psmallmatrix}$.\label{case:degree-3-2}
    \item If $\occ{F}{x}=\begin{psmallmatrix}
    x & \n{x} & \n{x}\\
    a & a & a\\
    z & y & \n{y}\\
    \end{psmallmatrix}$, then $\occ{F}{a}=\begin{psmallmatrix}
    x & \n{x} & \n{x} & b\\
    a & a & a & \n{a} \\
    z & y & \n{y} & z\\
    \end{psmallmatrix}$, and
	$\vdeg{b} = \vdeg{z} = 4$.
    \label{case:overlap}
    \end{enumerate}
\end{lemma}

\ifdefined\isarxiv
\begin{proof}
The proof heavily relies on the variable elimination operation as introduced in Section~\ref{sec:prelim}.

Let $F$ be an unsatisfiable $(3,4)$-formula of minimal size.
\begin{enumerate}
 \item If there is a literal $a$ that occurs in $F$ and $\n{a}$ does not occur in $F$, then we can obtain a smaller unsatisfiable $(3,4)$-formula by removing all clauses where $a$ occurs in $F$.
 \item Every clause in $F$ contains three distinct literals, so $a\not=b$. If $a=\n{b}$, then we can remove this clause and obtain a smaller unsatisfiable $(3,4)$-formula.
 \label{case:not-both-literals-in-same-clause}
 \item 
 Let $\occ{F}{y}=\begin{psmallmatrix}
    y & \overline{y}\\
    a & c\\
    b & d\\
\end{psmallmatrix}$. We need to show that $\{c,d\}=\{a,b\}$; since $a,b,c,d$ are symmetric, it suffices to show $c \in \{a,b\}$.
Suppose to the contrary that $c\not\in\{a,b\}$, the rest follows analogously.
	By eliminating $\var(y)$ we get the resolvent $\begin{psmallmatrix}
    c\\
    a\\
    b\\
    d\\
    \end{psmallmatrix}$, which is logically weaker than or equivalent to $\begin{psmallmatrix}
    c\\
    a\\
    b\\\end{psmallmatrix}$. By assumption, $a,b,c$ are distinct literals, so we can replace $\occ{F}{y}$ with $\begin{psmallmatrix}
    c\\
    a\\
    b\\
    \end{psmallmatrix}$ to obtain a smaller unsatisfiable $(3,4)$-formula. This contradicts our assumption that $F$ is an unsatisfiable $(3,4)$-formula of minimal size.
     \item It suffices to show that there is a literal $a$ that occurs in all three clauses of $F[x]$.
		 Suppose to the contrary that there is no such literal.
		 Let $\occ{F}{x}=\begin{psmallmatrix}
    x & \overline{x} & \overline{x}\\
    a & e & f\\
    b & c & d\\
    \end{psmallmatrix}$. Then $a$ does not occur in at least one of the last two clauses and $b$ does not occur in at least one of the last two clauses.  We distinguish the following two cases depending on whether there is a clause in $F[x]$ that neither $a$ nor $b$ occur in.
    \begin{enumerate}
    \item There is no such clause. Then $\{a,b\}\cap\{c,e\}\not=\emptyset$ and $\{a,b\}\cap\{d,f\}\not=\emptyset$. This means either $a\in\{a,b\}\cap\{d,f\}$ or $b\in\{a,b\}\cap\{d,f\}$. Suppose $a\in\{a,b\}\cap\{d,f\}$. Then $a\not\in\{a,b\}\cap\{c,e\}$. By assumption, $\{a,b\}\cap\{c,e\}\not=\emptyset$, which means $\{a,b\}\cap\{c,e\}=\{b\}$. This means that $b\not\in \{d,f\}$. Therefore, $\{a,b\}\cap\{d,f\}=\{a\}$.
		By eliminating $\var(x)$ we obtain $\begin{psmallmatrix}
    a & b\\
    e & f\\
    c & d\\
    \end{psmallmatrix}$. Since every clause in F needs to contain exactly 3 different literals, it follows that $e\not= c$ and $f\not=d$. This means that $a,c,e$ are distinct literals and $b,f,d$ are distinct and we can replace $\occ{F}{x}$ with $\begin{psmallmatrix}
    a & b\\
    e & f\\
    c & d\\
    \end{psmallmatrix}$ to obtain a smaller unsatisfiable $(3,4)$-formula. This contradicts our assumption that $F$ is an unsatisfiable $(3,4)$-formula of minimal size. Assuming $b\in\{a,b\}\cap\{d,f\}$ leads to contradiction by a similar argument.
    \item There is such a clause. Then $\{a,b\}\cap\{c,e\}=\emptyset$ or $\{a,b\}\cap\{d,f\}=\emptyset$. Suppose $\{a,b\}\cap\{c,e\}=\emptyset$. We distinguish the following three cases depending on the size of $\{a,b\}\cap\{f,d\}$. 
    \begin{enumerate}
		\item $|\{a,b\}\cap\{f,d\}| = 2$. Then $\{a,b\}=\{f,d\}$.
			By eliminating $\var(x)$, we obtain $\begin{psmallmatrix}
    a & a\\
    b & b\\
    e & \times\\
    c & \times \\
    \end{psmallmatrix}$, which is logically equivalent to $\begin{psmallmatrix}
    x & \overline{x}\\
    a & a\\
    b & b\\
    \end{psmallmatrix}$. Note that since every clause in $F$ contain three distinct literals, it follows that $x,\overline{x},a,b$ are all distinct literals. Then we can replace $\occ{F}{x}$ with $\begin{psmallmatrix}
    x & \overline{x}\\
    a & a\\
    b & b\\
    \end{psmallmatrix}$ to obtain a smaller unsatisfiable $(3,4)$-formula. This contradicts our assumption that $F$ is an unsatisfiable $(3,4)$-formula of minimal size. 
    \item $|\{a,b\}\cap\{f,d\}| = 1$. Without loss of generality, assume $f \neq a,b$ and $d = a$.
		By eliminating $\var(x)$ we obtain $\begin{psmallmatrix}
    a & a\\
    b & b\\
    e & f\\
    c & \times \\
    \end{psmallmatrix}$, which is logically weaker than $\begin{psmallmatrix}
    a & a\\
    e & f\\
    c & b\\
    \end{psmallmatrix}$. By assumption, $a,e,c$ and $a,f,b$ are two sets of distinct literals. This means that we can replace $F[x]$ with $\begin{psmallmatrix}
    a & a\\
    e & f\\
    c & b\\
    \end{psmallmatrix}$ to obtain a smaller unsatisfiable $(3,4)$-formula. This contradicts our assumption that $F$ is an unsatisfiable $(3,4)$-formula of minimal size.
    \item $|\{a,b\}\cap\{f,d\}| = 0$. Then $\{a,b\}\cap\{f,d\} = \emptyset$.
		By eliminating $\var(x)$ we obtain $\begin{psmallmatrix}
    a & a\\
    b & b\\
    e & f\\
    c & d \\
    \end{psmallmatrix}$, which is logically weaker than $\begin{psmallmatrix}
    a & b\\
    e & f\\
    c & d\\
    \end{psmallmatrix}$. By assumption, $a,e,c$ and $b,f,d$ are two sets of distinct literals. This means that we can replace $F[x]$ with  $\begin{psmallmatrix}
    a & b\\
    e & f\\
    c & d\\
    \end{psmallmatrix}$ to obtain a smaller unsatisfiable $(3,4)$-formula. This contradicts our assumption that $F$ is an unsatisfiable $(3,4)$-formula of minimal size.
    \end{enumerate}
    \end{enumerate}
    Because all cases run into contradiction, we conclude that there is a literal $a$ that occurs in all three clauses of $F[x]$.

\item  Suppose $\occ{F}{x}\cap\occ{F}{w}\not=\emptyset$. We distinguish the following two cases depending on whether $F[x]=F[w]$.
    \begin{enumerate}
    \item $\occ{F}{x}=\occ{F}{w}$. Since no two clauses can contain exactly the same literals, $\occ{F}{x}=\occ{F}{w}$ has to be of the form $\begin{psmallmatrix}
    x & \overline{x} & \overline{x}\\
    a & a & a\\
    \overline{w} & w & \overline{w}\\
\end{psmallmatrix}$. 

    Then the first clause is blocked on $\n{w}$ and can therefore be removed from $F$ to obtain a smaller unsatisfiable $(3,4)$-formula. This contradicts our assumption that $F$ is an unsatisfiable $(3,4)$-formula of minimal size.
    \item $\occ{F}{x}\not=\occ{F}{w}$. This means that $|\occ{F}{x}\cup\occ{F}{w}| > 3$. Let $\begin{psmallmatrix}
    a\\
    b \\
    c\\
    \end{psmallmatrix}\in \occ{F}{x}\cap\occ{F}{w}$. It is safe to assume $a\in\{x,\overline{x}\}$ and $b\in\{w,\overline{w}\}$. By fact~\ref{case:var-deg-3}, it follows that one literal from $\{a,b,c\}$ occur in all three clauses of $F[x]$. By fact~\ref{case:no-pure}, $\ldeg{a},\ldeg{b}\in\{1,2\}$. This means that $c$ occur in all clauses of $F[x]$. Similarly, we can argue that $c$ occur in all clauses of $F[w]$. This means that $\ldeg{c}\geq |\occ{F}{x}\cup\occ{F}{w}| > 3$. This contradicts the fact that given there is no pure literals in $F$, every literal can occur at most three times.
    \end{enumerate}
    
\item First, we show that either $F[x]\cap F[y]=\emptyset$ or $F[y]\subseteq F[x]$. Suppose to the contrary that $F[x]\cap F[y]\not=\emptyset$ and $F[y]\not\subseteq F[x]$. Then $|F[x]\cup F[y]| > 3$. Let $\begin{psmallmatrix}
    a\\
    b \\
    c\\
	\end{psmallmatrix}\in \occ{F}{x}\cap\occ{F}{y}$. It is safe to assume $a\in\{x,\overline{x}\}$ and $b\in\{y,\overline{y}\}$. By fact~\ref{case:var-deg-3}, it follows that one literal from $\{a,b,c\}$ occurs in all three clauses of $F[x]$. By fact~\ref{case:no-pure}, $\ldeg{a},\ldeg{b}\in\{1,2\}$. This means that $c$ occurs in all clauses of $F[x]$. By fact~\ref{case:var-deg-3}, $c$ occurs in all clauses of $F[y]$. This means that $\ldeg{c}\geq |\occ{F}{x}\cup\occ{F}{y}| > 3$. But then $\ldeg{c} = 4$, and so $\ldeg{\n{c}} = 0$, a contradiction with fact~\ref{case:no-pure}.

	So, either $F[x]\cap F[y]=\emptyset$ or $F[y]\subseteq F[x]$. Now suppose $F[y]\subseteq F[x]$. By facts~\ref{case:var-deg-2} and \ref{case:var-deg-3}, $\occ{F}{x}$ is of the form $\begin{psmallmatrix}
    x & \overline{x} & \overline{x}\\
    a & a & a\\
    z & y & \overline{y}\\
    \end{psmallmatrix}$. Thus, either $\occ{F}{x}\cap\occ{F}{y}=\emptyset$, or $\occ{F}{x}$ is of the form $\begin{psmallmatrix}
    x & \n{x} & \n{x}\\
    a & a & a\\
    z & y & \n{y}\\
    \end{psmallmatrix}$.
    
\item Suppose $\occ{F}{x}=\begin{psmallmatrix}
    x & \n{x} & \n{x}\\
    a & a & a\\
    z & y & \n{y}\\
    \end{psmallmatrix}$. Note that since $a$ occurs three times in $F[x]$, there has to be a unique clause in $F$ where $\overline{a}$ occurs. Let this clause be  $\begin{psmallmatrix}
    \overline{a}\\
    b\\
    c\\
    \end{psmallmatrix}$. We can eliminate $\var(y)$ to obtain $\begin{psmallmatrix}
    x &\n{x} &\n{a} \\
    a & a & b \\
    z & \times & c \\
\end{psmallmatrix}$, then eliminate $\var(x)$ to obtain $\begin{psmallmatrix}
    z &\n{a} \\
    a & b \\
    \times & c \\
\end{psmallmatrix}$, and finally eliminate $\var(a)$ to obtain $\begin{psmallmatrix}
    z \\
    b \\
    c \\
\end{psmallmatrix}$. By assumption, $b\not= c$. If $z\not=b,c$, then we can replace $\occ{F}{a}$ with$\begin{psmallmatrix}
    z \\
    b \\
    c \\
\end{psmallmatrix}$ to obtain a smaller unsatisfiable $(3,4)$-formula, which contradicts our assumption that $F$ is an unsatisfiable $(3,4)$-formula of minimal size. Thus,  $z=b$ or $z=c$. In other words, $\occ{F}{a}$ is of the form $\begin{psmallmatrix}
    x & \n{x} & \n{x} & b\\
    a & a & a & \n{a} \\
    z & y & \n{y} & z\\
\end{psmallmatrix}$. To see that $\vdeg{z} = 4$, note that $\vdeg{z} > \ldeg{z} \geq 2$, and if $\vdeg{z} = 3$, then $F[z]\cap F[x]\not=\emptyset$, which contradicts fact~\ref{case:var-deg-3}.
    To see that $\vdeg{b} = 4$, assume to the contrary that either $\vdeg{b} = 2$ or $\vdeg{b} = 3$:
    \begin{enumerate}
    \item Suppose $\vdeg{b} = 2$. By fact~\ref{case:var-deg-2}, $F[b] = \begin{psmallmatrix}
    b & \n{b}\\
    \n{a} & \n{a} \\
    z & z\\
    \end{psmallmatrix}$. This means that $\vdeg{a}\geq 3 + 2 = 5$, which is impossible in a $(3,4)$-formula.
\item Suppose $\vdeg{b} = 3$. By fact~\ref{case:var-deg-3}, either $\n{a}$ or $z$ occurs three times in $F[b]$. If $\n{a}$ occurs three times in $F[b]$, then $\vdeg{a}\geq 3 + 3 = 6$, which is again impossible. If $z$ occurs three times in $F[b]$, then since $F[x]\cap F[b] = \emptyset$ and $z$ occurs once in $F[x]$, it follows that $\ldeg{z}\geq 3 + 1 = 4$, and thus $\ldeg{\n{z}} = 0$, which contradicts fact~\ref{case:no-pure}.
    \end{enumerate}
    Because both possibilities run into contradiction, we conclude that $\vdeg{b}=4.$
\end{enumerate}
One can argue that these facts also hold for an unsatisfiable $(3,2,2)$-formula of minimal size with the same arguments. This completes the proof of the lemma.
\end{proof}
\fi

\begin{corollary}
\label{cor:kst-deg3-var} A smallest $(3, 2, 2)$-formula has even size and no variables of degree $3$. 
\end{corollary}
\begin{proof}
By fact \ref{case:var-deg-3}\ from \Cref{lemma:minimal-formula}, the existence of a variable $x$ of degree $3$ implies that there is a literal $a$ of degree $3$; a contradiction in a $(3,2,2)$-formula. So, a minimal $(3,2,2)$-formula contains only variables of type $(1,1)$ and $(2,2)$, and thus has an even number of literal occurrences. With $k=3$, the number of clauses must be even.
\end{proof}

One caveat to searching for a formula of a reduced profile using SMS is that the reduction can reduce degrees of some variables and literals. This means that some of the variables or literals in the shorter clauses must occur strictly fewer times than the general bound $s$ (or \mbox{$p$, $q$)}. To accommodate for this, we count each literal that occurs in a $2$-clause twice. However, this means that the literals in the $2$-clauses whose degree was not reduced by the steps above may exceed their degree cap under this way of counting. To adjust to this, we sum up the number of exceeding counts and call it the \emph{surplus} of the formula. Given a specific reduction, we know the exact value of allowed surplus for the reduced formula, and so we can include it as a part of the constraints. Here as an example, we write out the formula for the case of bounded variable degree. The formula for the literal case can be defined similarly.
For all $i\in[n]\cup[-n],j\in[m]$, define $\incip{i}{j}:=\inci{i}{j}\land\lnot\tau_{j,k}$. For all $i\in[m]$ and $t\in[m]\cup[-m]$, define $(\sigma'_{i,1},\sigma'_{i,2},\dots,\sigma'_{i,s+s}):=\text{Card}(\sigma_{i,1},\sigma_{i,2},\dots,\sigma_{i,s},\incip{i}{1},\dots,\incip{i}{m},\incip{-i}{1},\dots,\incip{-i}{m})$. Suppose $\surplus$ is the value of surplus. Then we can define 
\[F_{\text{surp}}:=\sum_{i\in[n]}\sigma'_{i,s+1}+\dots+\sigma'_{i,s+s}=\mathcal{S}.\]

We call the combination of the number of variables $n$, the number of clauses $m$, a stairway~$\sigma$ and the number of surplus $\surplus$ a \emph{profile}. For further speed-up, we also hard-code a number of variables of degree $3$ in the same way as described in the previous part. Let $F$ be a smallest unsatisfiable $(3,4)$-formula and let $x_1,x_2,\dots,x_{\fixvar}$ be $\fixvar$ degree $3$ variables in~$F$ whose occurrences we want to fix. By \Cref{lemma:minimal-formula}, the formula is of the following form. We fix the first $3\fixvar$ rows of the matrix thereby determined.
\[\begin{psmallmatrix}
x_1 & \overline{x}_1 & \overline{x}_1\\
a_1 & a_1 & a_1\\
 & & \\
\end{psmallmatrix}\begin{psmallmatrix}
x_2 & \overline{x}_2 & \overline{x}_2\\
a_2 & a_2 & a_2\\
& & \\
\end{psmallmatrix}\dots\begin{psmallmatrix}
x_{\fixvar} & \overline{x}_{\fixvar} & \overline{x}_{\fixvar}\\
a_{\fixvar} & a_{\fixvar} & a_{\fixvar}\\
&  &\\
\end{psmallmatrix}\dots\]

\begin{lemma}
  \label{lemma:3-4-profiles}
  A smallest unsatisfiable $(3,4)$-formula (or $(3,2,2)$-formula) with $n$ variables and $m$
  clauses exists only if a formula of one of the profiles in \Cref{table:3-4-cases}, top (bottom) exists.
\end{lemma}

\begin{proof}
	Let $F$ be an unsatisfiable $(3,4)$-formula of minimal size. By~\Cref{lemma:minimal-formula}, $F$ must contain at least one clause where all variables are of degree $4$; otherwise each clause is contained in some $\occ{F}{x}$ for a variable $x$ with $\deg(x) \in \{2,3\}$, and setting the respective literals $a$ (from facts \ref{case:var-deg-2} and~\ref{case:var-deg-3}) to true satisfies $F$. Let $n_2$ be the number of variables of degree $2$ in $F$ and $n_3$ be the number of variables of degree $3$ in $F$.
	By Lemma~\ref{lem:at-most-one-deg-2}, $n_2 \in \{0, 1\}$, and counting the number of literal occurrences we get $n_3 = 4n - 3m -2n_2$.
	This means that for a given pair $(n,m)$, there are at most 3 possibilities to consider: $n_2=0$ or $n_2=1$, and if $n_2=1$, whether the clauses of the variable of degree~$2$ overlap with the clauses of some variable of degree~$3$. The following considerations further restrict the possibilities.

\begin{enumerate}
\item if there is a variable of degree $2$ and its clauses overlap with the clauses of variables of degree $3$, then by~\Cref{lemma:minimal-formula}, facts~\ref{case:disjoint-deg-3}, \ref{case:degree-3-2} and \ref{case:overlap}, each clause contains at most one variable of degree $3$ and there is at least one clause where every variable has degree $4$. Therefore, $m\geq 3n_3 + 1$;
\item otherwise, by facts~\ref{case:disjoint-deg-3} and \ref{case:degree-3-2}, each clause contains at most one variable of degree less than~$4$. Also, from the reasoning above, there has to be a clause with only variables of degree less than~$4$. This means that $m\geq 3n_3 + 2n_2 + 1$.
\end{enumerate}

This leaves a set of combinations of $(n,m)$ and $(n_2,n_3)$ for the unsolved cases from Table~\ref{table:3-4-time} that we list in Table~\ref{table:3-4-cases}, along with the profiles resulting from the variable-elimination-based replacement as described above \Cref{lem:at-most-one-deg-2}
It is readily verified that since both replacement patterns require a variable of degree~$2$, and neither creates new variables of degree~$2$, at most one replacement step is applicable.
\end{proof}

\begin{table}[tbp]
  \caption{Time spent determining the existence of a formula with each profile from reduction for
	  previously unsolved cases with $n$ variables and $m$ clauses from Table~\ref{table:3-4-time} (c.f.\ the proof of \Cref{lemma:3-4-profiles}; $n',m'$ is the number of variables and clauses after replacement, which reflects the values $n_2,n_3$ and the structure of the replacement operation); $\fixvar$ is the number of hard-coded variables of degree $3$ we chose (naturally, at most the total number of variables of degree $3$ after replacement).
	In the bottom table, the number of hard-coded variables of degree $3$ is always $0$ because by
    Corollary \ref{cor:kst-deg3-var} no variable of degree~$3$ exists.
  }
  \label{table:3-4-cases}
  \av{\medskip}
\centering
     \setlength\tabcolsep{2.5mm}
  \begin{tabular}{@{}cc@{~~}ccc@{~~}rr@{~~}ccc@{~~}c@{}}
	\multicolumn{10}{@{}l}{profiles for $(3,4)$-formulas} \\
\toprule
$n$&$m$&$n_2$&$n_3$&overlap?&~~$n'$&$m'$&$\#2$-cl&$\surplus$&$\fixvar$&time\\
\midrule
12 & 14 &1 &4&yes& 9 & 11 & 1 & 1 & 3 & 0.2s\\
\midrule
&&0 &3&-& 12 & 15 & 0 & 0 & 3 & 6.5s\\
12 & 15 &1&1&no& 11 & 14 & 1 & 0 & 1 & 2.1h\\
&&1&1&yes& 9 & 12 & 1 & 1 & 0 & 41m\\
\bottomrule
\end{tabular}\hfill%
\\[2em]
\begin{tabular}{@{}ccccccccc@{}}
	\multicolumn{9}{@{}l}{profiles for $(3,2,2)$-formulas} \\
\toprule
$n$&$m$&$n_2$&$n_3$&$n'$&$m'$& $\#2$-cl& $\surplus$&time\\
\midrule
12 & 16 &0&0&12 & 16 & 0 & 0 & 5d\\
14 & 18 &1&0&13 & 17 & 1 & 0 & 21d\\
\bottomrule
\end{tabular}
\end{table}

We tested the existence of each of the profiles from Table~\ref{table:3-4-cases} with SMS and all of them returned negative. Combining this computational result with Lemma~\ref{lemma:3-4-profiles}, we obtain the following theorem.

\begin{theorem}
$\mu(3,4)> 15$ and  $\mu(3,2,2)>18$.
\end{theorem}

\section{A smaller unsatisfiable \texorpdfstring{$(4,5)$}{(4,5)}-formula}
\label{sec:smaller_4_5}
We now turn our attention to unsatisfiable CNF formulas with clauses
of length $4$.  As mentioned in the introduction, $f(4)=4$, i.e.,
$s=5$ is the smallest value so that an unsatisfiable $(4,s)$\hy
formula exists.  We found an unsatisfiable $(4,s)$\hy formula with 235
clauses, improving upon the formulas provided by St\v{r}\'{\i}brn\'{a}
\cite{Stribrna94} (449 clauses) and Knuth \cite[p.~588]{Knuth23} (257
clauses).
Knuth's formula $K_{4,5}$ can be described by
the $\oplus$\hy derivation from the axiom $\{\nil\}$ depicted as a directed acyclic graph in Figure~\ref{fig:K_4_5_derivation}.

\begin{figure}[h]
  \centering
  \tikzstyle{tightcircle}=[circle,draw=black,inner sep=0.75pt]

  \av{	  \begin{tikzpicture}[scale=1.2]}
  \cv{	  \begin{tikzpicture}}
    
		  \node[draw=black] (0) at (0, 1) { };

		  \node[tightcircle] (1) at (1, 1) {$F_1$};
		  \node[tightcircle] (2) at (2, 1) {$F_2$};
		  \node[tightcircle] (3) at (2, 0) {$F_3$};
		  \node[tightcircle] (4) at (3, 0) {$F_4$};
		  \node[tightcircle] (5) at (3, 1) {$F_5$};
		  \node[tightcircle] (6) at (3, 2) {$F_6$};
		  \node[tightcircle] (7) at (4, 2) {$F_7$};
		  \node[tightcircle] (8) at (4, 1) {$F_8$};
		  \node[tightcircle] (9) at (4, 0) {$F_9$};

		  \newcommand{\dy}{4}
		  \node[tightcircle] (10) at (5, \dy-3) {$F_{10}$};
		  \node[tightcircle] (11) at (6, \dy-2) {$F_{11}$};
		  \node[tightcircle] (12) at (7, \dy-2) {$F_{12}$};

		  \node[tightcircle] (13) at (7, \dy-4) {$F_{13}$};
		  \node[tightcircle] (14) at (8, \dy-3) {$F_{14}$};
		  \node[tightcircle] (15) at (9, \dy-2) {$F_{15}$};

		  \node[tightcircle] (16) at (9, \dy-4) {$F_{16}$};
		  \node[tightcircle] (17) at (9.5, \dy-3) {$F_{17}$};

		  \node[circle,draw=black,inner sep=0.1pt] (K)  at (11, \dy-3) {$K_{4,5}$};

		  \draw[->>] (1) -- (2);
		  \draw[->>>] (2) -- (3);
		  \draw[->>>>] (3) -- (4);
		  \draw[->>>] (4) -- (5);
		  \draw[->>>>] (5) -- (6);
		  \draw[->>>>] (6) -- (7);
		  \draw[->>] (7) -- (8);
		  \draw[->>>>] (8) -- (9);

		  \draw[->>>] (9) -- (10);
		  \draw[->>>] (10) -- (11);
		  \draw[->>>>] (11) -- (12);

		  \draw[->>>] (9) -- (13);
		  \draw[->>>] (10) -- (14);

		  \draw[->>] (12) -- (13);
		  \draw[->>] (12) -- (14);
		  \draw[->>] (12) -- (15);

		  \draw[->>>] (14) -- (15);

		  \draw[->>] (13) -- (16);
		  \draw[->>>] (14) -- (16);
		  \draw[->>>] (14) -- (17);
		  \draw[->>] (15) -- (17);

		  \draw[->>>] (16) -- (K);
		  \draw[->>] (17) -- (K);

		  \newcommand{\rot}{90}
		  \draw[->>] (1) to[out=\rot-45,in=\rot+45] (5);
		  \draw[->>>] (4) -- (8);
		  \draw[->>] (7) -- (10);
		  \draw[->>] (7) -- (11);
    
		  \draw[->] (0) to[out=\rot-75 ,in=\rot+75] (1);
		  \draw[->] (0) to[out=\rot-105,in=\rot+105] (1);
		  \draw[->] (0) to[out=\rot-120,in=\rot+120] (2);
		  \draw[->] (0) to[out=\rot-135,in=\rot+90] (3);
		  \draw[->] (0) to[out=\rot-150,in=\rot+120] (4);
		  \draw[->] (0) to[out=\rot-60 ,in=\rot+90] (6);
		  \draw[->] (0) to[out=\rot-45 ,in=\rot+60] (7);
		  \draw[->] (0) to[out=-75,in=-150] (9);
		  \draw[->] (0) to[out=60,in=150] (12);
	  \end{tikzpicture}
	  \caption{The $\oplus$\hy derivation of Knuth's formula $K_{4,5}$.
		  Each node has two incoming arcs; the number of arrowheads denotes the values $p,q$ in $\oplus_{p,q}$.
	  }
	  \label{fig:K_4_5_derivation}
  \end{figure}
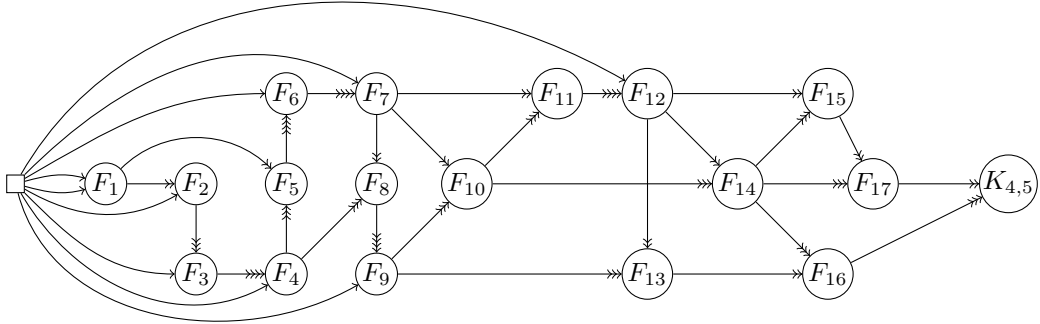
This formula has 257 clauses and 256 variables. However, 22 of these variables occur twice (introduced in $F_1$) and 31 variables occur in three clauses (introduced in $F_2$). Hence we can identify each of the variables with 2 occurrences with one variable with 3 occurrences, still keeping the number of occurrences of the new identified variable within the bound 5. By this method, we can save 22 variables, and indeed, Knuth states his formula to have $256-22=234$ variables. 
With our modified saturation algorithm we could show that $K_{4,5}$ is a smallest unsatisfiable $(4,5)$\hy formula
in $\MU(1)$. Thus, for finding a smaller formula, one needs to search
outside $\MU(1)$, and outside the class of formulas that can be obtained from an $\MU(1)$ formula by identifying pairs of low-occurrence variables.

Finding an entire unsatisfiable $(4,5)$\hy formula with SMS does not
seem feasible. However, we can search for an unsatisfiable $(\leq
4,5)$\hy formula $F$ that represents a stairway $\sigma$ with fewer
clauses than any formula in $\MU(1)$ that represents the same stairway
$\sigma$. Thus way we can use $F$ as an additional axiom in
$\oplus$\hy derivations and this way possibly find a smaller unsatisfiable $(4,5)$\hy formula.
We considered all stairways  $\sigma\in \{(3)$, $(2)$, $(3,2)$, $(2,2)$, $(2,2,2)$, $(1)$, $(3,1)$, $(2,1)$,
$(2,2,1)$, $(1,1)$, $(3,1,1)$, $(2,1,1)$, $(2,2,1,1)$, $(1,1,1)$,
$(3,1,1,1)$, $(2,1,1,1)$, $(1,1,1,1)\}$ and run our SMS encoding to
find an unsatisfiable  formula $F$ with at most 13 clauses and fewer clauses 
than a smallest $\MU(1)$\hy formula for $\sigma$.

The search resulted in two such formulas within a timeout of five
days. One formula has $7$ clauses for the stairway $(3,1,1,1)$ and the
other has $8$ clauses for the stairway $(3,2)$; shortest $\MU(1)$
formulas for these stairways have $8$ and $9$ clauses,
respectively. Using the first of these two formulas as axiom indeed
reduces the size of the unsatisfiable $(4,5)$\hy formula from 257 to
235, since the axiom is used several times; the second formula can be
derived by an $\oplus$ operation from the new axiom and axiom
$\{\nil\}$.  Our smaller unsatisfiable $(4,5)$\hy formula can be obtained by replacing $F_6$
with $F_6'$
in the $\oplus$\hy derivation from Figure~\ref{fig:K_4_5_derivation},
and using $F_6'$ as an additional axiom, where the two formulas are as follows (rows are variables, columns are clauses, \plus/\minus\ indicates positive/negative occurrence; c.f.\ the appendix).
\renewcommand\arraystretch{0.4}
\newcommand{\cwidth}{0.257cm}
\setlength{\tabcolsep}{1.7pt}
\[F_6'=\left\lgroup\begin{tabular}{p{\cwidth}p{\cwidth}p{\cwidth}p{\cwidth}p{\cwidth}p{\cwidth}p{\cwidth}}
\minus & & & & & &\plus \\
&\minus &\minus &\minus &\minus &\plus & \\
\minus & &\minus &\plus &\plus & &\minus \\
\minus &\minus &\plus & &\plus & &\minus \\
\minus &\plus & &\minus &\plus & &\minus \\
\end{tabular}\right\rgroup
\quad
F_6=\left\lgroup\begin{tabular}{p{\cwidth}p{\cwidth}p{\cwidth}p{\cwidth}p{\cwidth}p{\cwidth}p{\cwidth}p{\cwidth}}
	   &       &       &\minus &\minus & \minus &\minus & \plus \\
       &       &\minus &\minus &\minus & \plus  &\plus  &       \\
\minus &\minus &\minus &\minus &\plus  &        &       &       \\
\minus &\minus &\minus &\plus  &       &        &       &       \\
\minus &\minus &\plus  &       &       &        &       &       \\
\minus &\plus  &       &       &       & \minus &\plus  &       \\
\end{tabular}\right\rgroup
\]

\section{Conclusion}
We have identified 
the smallest unsatisfiable $(3, s)$ and $(3,p,q)$\hy formulas for a
comprehensive range of values, and brought an improvement in the known
minimal size for an unsatisfiable  $(4, 5)$-CNF formula. 
Our work also contributed a uniform proof of the dichotomy for $(3,p,q)$-SAT.
The core methodology, a fusion of theoretical insights and an innovative application of the SMS framework with the methods of disjunctive splitting and reductions has not only led to the discovery of new smallest unsatisfiable formulas but also demonstrated the practical utility of SMS in exploring the combinatorial landscape of CNF formulas.

Our findings have
illuminated several challenging and open avenues for future work. For
instance, extending the scope to identify smallest unsatisfiable
formulas for $k \geq 4$ remains a significant challenge. While the
methods developed here provide a solid foundation, both novel
techniques and theoretical advances are necessary to tackle the
increased complexity of larger $k$ values. An extension of our methods
may lead to determining the exact value of the threshold $f(5)$,
currently only known to be in the interval $[5,7]$. Moreover, the
interplay between the size of unsatisfiable formulas and their
implications for inapproximability results in MaxSAT problems~\cite{BermanKarpinskiScott03}  invites
deeper investigation. Finally, since we found out that, for all $q\geq 3$, the size of the
smallest unsatisfiable $(3,1,q)$-formula coincide with that of the smallest unsatisfiable $(3,1,q)$-formula in $\MU(1)$,
we conjecture that this is true for all $k\geq 3$ and $q\geq f(k)+1$.

\av{%
\section*{Acknowledgements}
\nopagebreak
\begin{minipage}{.88\linewidth}
 The research leading to this publication has received funding from
 the European Union's Horizon 2020 research and innovation programme
 under grant agreement No.~101034440, and was supported by the
 Austrian Science Fund (FWF) within the project \fwf{P36688}.
\end{minipage}%
\begin{minipage}[t]{.12\linewidth}
  \hfill\includegraphics[width=15mm]{eu_logo.pdf}
\end{minipage}
}

\appendix

\section{Formulas}
In the appendix we list formulas in tabular form, where rows represent variables, columns represent clauses, and the \plus\ and \minus\ signs indicate positive and negative occurrence, respectively.
Blank cells mean the variable does not occur in the clause.
\subsection{Smallest enforcers}
\label{appendix:enforcers}
In the following, we give examples of smallest $(3,s)$\hy enforcers $E_{3,s}$ and smallest $(3,p,q)$\hy enforcers $E_{3,p,q}$ for all entries in Table \ref{tab:3-enf}. Note that $E_{3,4}$ also acts as a smallest $(3,2,3)$-enforcer.

\vspace{1em}
\noindent
\;
$E_{3,5}= \left\lgroup\begin{tabular}{@{}p{\cwidth}p{\cwidth}p{\cwidth}p{\cwidth}}
\minus &\minus &\minus &\minus \\
\minus &\minus &\plus  &\plus  \\
\minus &\plus  &\minus &\plus  \\
\end{tabular}\right\rgroup$
$\quad\quad\quad\,
E_{3,4}= \left\lgroup\begin{tabular}{@{}p{\cwidth}p{\cwidth}p{\cwidth}p{\cwidth}p{\cwidth}}
       & \minus &\minus &\minus &       \\
\minus &        &\minus &\plus  & \plus \\
\minus &\minus  &\plus  &       & \plus \\
\minus &\plus   &       &\minus & \plus \\
\end{tabular}\right\rgroup$
\quad\quad\,
$E_{3,1,5}= \left\lgroup\begin{tabular}{@{}p{\cwidth}p{\cwidth}p{\cwidth}p{\cwidth}p{\cwidth}}
\minus & \minus & \minus & \minus & \minus \\
\minus &        &        & \minus & \plus  \\
       & \minus & \minus &        & \plus  \\
\minus &        &        & \plus  &        \\
       & \minus & \plus  &        &        \\
\end{tabular}\right\rgroup$
\\[3pt]
$E_{3,1,4}= \left\lgroup\begin{tabular}{@{}p{\cwidth}p{\cwidth}p{\cwidth}p{\cwidth}p{\cwidth}p{\cwidth}}
\minus &        &        &         & \minus & \plus \\
       & \minus & \minus & \minus  &        & \plus  \\
\minus &        &        & \minus  & \minus &\minus  \\
\minus &        &        &         & \plus  &        \\
       & \minus & \minus & \plus  &         &        \\
       & \minus & \plus  &        &         &        \\
\end{tabular}\right\rgroup$
\,
$E_{3,1,3}= \left\lgroup\begin{tabular}{@{}p{\cwidth}p{\cwidth}p{\cwidth}p{\cwidth}p{\cwidth}p{\cwidth}p{\cwidth}}
       &        &        &        & \plus  & \plus  & \plus \\
\minus &        &        & \minus &        & \minus & \plus \\
       & \minus & \minus &        & \minus &        & \plus \\
\minus &        &        & \minus &        & \plus  &       \\
       & \minus & \minus &        & \plus  &        &       \\
\minus &        &        & \plus  &        &        &       \\
       & \minus & \plus  &        &        &        &       \\
\end{tabular}\right\rgroup$
\,
$E_{3,2,2}= \left\lgroup\begin{tabular}{@{}p{\cwidth}p{\cwidth}p{\cwidth}p{\cwidth}p{\cwidth}p{\cwidth}p{\cwidth}p{\cwidth}p{\cwidth}p{\cwidth}}
\minus &        & \minus &        &        &        &        &        &        &        \\
       & \minus &        &        &        & \minus &        &        & \plus  & \plus  \\
       &        &        &        & \minus &        & \minus & \plus  &        & \plus  \\
\minus &        &        & \minus &        &        & \plus  &        &        & \plus  \\
       &        &        &        & \minus & \plus  &        & \minus & \plus  &        \\
       & \minus & \minus &        &        & \plus  &        & \plus  &        &        \\
       &        &        & \minus & \plus  &        & \minus &        & \plus  &        \\
\minus & \minus & \plus  & \plus  &        &        &        &        &        &        \\
\end{tabular}\right\rgroup$

\subsection{Smallest MU(1) formulas}
\setlength{\tabcolsep}{4.5pt}
In the following, we give $\oplus$\hy derivations of the smallest $(3,s)$\hy formula $M^1_{3,s}$
and the smallest $(3,p,q)$\hy formula $M^1_{3,p,q}$ restricted to
$\MU(1)$ for each entry from Table~\ref{table:smallest-MU-1}. For all derivations, we have \drule[1,1]{1}{0}{0}, \drule[1,2]{2}{0}{1} and \drule[2,2]{2'}{1}{1}; other 
$F_i$ symbols are defined locally in each derivation.

\begin{longtable}
{@{} l l l l l@{}}
\toprule
\Mdrule[2,2]{3,4}{4}{4} &\drule[1,3]{3}{0}{2} &\drule[2,2]{4}{3}{3}\\ 
\midrule
\Mdrule[3,2]{3,5}{4}{5} &\drule[1,3]{3}{0}{2} &\drule[2,3]{4}{1}{2} &\drule[3,2]{5}{2}{3}\\
\midrule
\Mdrule[4,2]{3,6}{2'}{3} &\drule[2,4]{3}{1}{2'}\\
\midrule
\Mdrule[4,3]{3,7}{2'}{3} &\drule[2,3]{3}{1}{2}\\
\midrule
\Mdrule[4,4]{3,8}{2'}{2'}\\
\midrule
\Mdrule[1,3]{3,1,3}{6}{4} &\drule[1,3]{3}{0}{2} &\drule[1,3]{4}{3}{2} &\drule[1,3]{5}{0}{4} &\drule[1,3]{6}{5}{4}\\
\midrule
\Mdrule[1,3]{3,1,4}{5}{3} &\drule[1,3]{3}{2}{2} &\drule[1,4]{4}{0}{3} &\drule[1,3]{5}{4}{3}\\
\midrule
\Mdrule[1,5]{3,1,5}{5}{3} &\drule[1,3]{3}{1}{2} &\drule[1,5]{4}{0}{3}&\drule[1,5]{5}{4}{3}\\
\midrule
\Mdrule[2,3]{3,2,3}{5}{4} &\drule[1,3]{3}{0}{2} &\drule[2,3]{4}{1}{2}&\drule[2,3]{5}{3}{2}\\
\midrule
\Mdrule[2,4]{3,2,4}{3}{2'} &\drule[2,4]{3}{1}{2'}\\
\midrule
\Mdrule[3,3]{3,3,3}{3}{3} &\drule[2,3]{3}{1}{2}\\
\midrule
\Mdrule[3,4]{3,3,4}{3}{2'} &\drule[2,3]{3}{1}{2}\\
\midrule
\Mdrule[4,4]{3,4,4}{2'}{2'}\\
\bottomrule
\end{longtable}

\subsection{Smallest \texorpdfstring{$(3,s)$}{(3,s)} and \texorpdfstring{$(3,p,q)$}{(3,p,q)}-formulas}
In the following, we give examples of smallest $(3,s)$\hy formulas $M_{3,s}$ and smallest $(3,p,q)$\hy formulas $M_{3,p,q}$ for all entries in Table \ref{tab:3-s}. We take $M_{3,4}=M^1_{3,4}$ and $M_{3,1,q}=M^1_{3,1,q}$ for all applicable $q$.

\noindent
\setlength{\tabcolsep}{1.7pt}
$M_{3,6}=M_{3,3,3}=\left\lgroup\begin{tabular}{@{}p{\cwidth}p{\cwidth}p{\cwidth}p{\cwidth}p{\cwidth}p{\cwidth}p{\cwidth}p{\cwidth}}
&\minus & &\minus &\minus &\plus &\plus &\plus \\
\minus& &\minus & &\plus &\minus &\plus &\plus \\
\minus &\minus &\plus &\plus & &\minus & & \plus\\
\minus &\plus &\minus &\plus &\minus &\plus & &\\
\end{tabular}\right\rgroup$
\;\,
$M_{3,5}=M_{3,2,3}=\left\lgroup\begin{tabular}{@{}p{\cwidth}p{\cwidth}p{\cwidth}p{\cwidth}p{\cwidth}p{\cwidth}p{\cwidth}p{\cwidth}p{\cwidth}p{\cwidth}p{\cwidth}}
& & \minus& &\minus & & &\minus & &\plus &\plus \\
&\minus & &\minus & & &\minus & &\plus & &\plus \\
\minus& & & & &\minus & &\plus &\minus & &\plus \\
\minus &\minus & &\minus & & & \plus & & &\plus & \\
\minus & & & & &\minus &\plus & & \plus &\minus & \\
& &\minus & &\minus &\plus & &\plus & & & \\
&\minus &\minus & \plus&\plus & & & & & & \\
\end{tabular}\right\rgroup$
\\[1em]
\begin{tabular}{r c}
&
\multirow{3}{*}{$\left\lgroup\begin{tabular}{@{}p{\cwidth}p{\cwidth}p{\cwidth}p{\cwidth}p{\cwidth}p{\cwidth}p{\cwidth}p{\cwidth}p{\cwidth}p{\cwidth}p{\cwidth}p{\cwidth}p{\cwidth}p{\cwidth}p{\cwidth}p{\cwidth}p{\cwidth}p{\cwidth}p{\cwidth}p{\cwidth}}
\plus & & & & &\plus & & & & &\minus & & & & &\minus & & & & \\
&\plus &\plus & & & &\minus &\minus & & & & & & & & & & & & \\
& & & & & & & & & & &\plus &\plus & & & &\minus &\minus & & \\
& & & & & & & & & & & & & & &\minus & &\plus &\plus &\minus \\
& & & & & & & & & & & & & & &\plus &\minus & &\plus &\minus \\
& & & & & & & & & & & & & & & &\plus &\minus &\plus &\minus \\
& & & & & & & & & &\minus & &\plus &\plus &\minus & & & & & \\
& & & & & & & & & &\plus &\minus & &\plus &\minus & & & & & \\
& & & & & & & & & & &\plus &\minus &\plus &\minus & & & & & \\
& & & & &\minus & &\plus &\plus &\minus & & & & & & & & & & \\
& & & & &\plus &\minus & &\plus &\minus & & & & & & & & & & \\
& & & & & &\plus &\minus &\plus &\minus & & & & & & & & & & \\
\minus & &\plus &\plus &\minus & & & & & & & & & & & & & & & \\
\plus &\minus & &\plus &\minus & & & & & & & & & & & & & & & \\
&\plus &\minus &\plus &\minus & & & & & & & & & & & & & & & \\
\end{tabular}\right\rgroup$}\\
$M_{3,2,4}=\left\lgroup\begin{tabular}{@{}p{\cwidth}p{\cwidth}p{\cwidth}p{\cwidth}p{\cwidth}p{\cwidth}p{\cwidth}p{\cwidth}p{\cwidth}p{\cwidth}}
& & & &\minus & & & & &\plus \\
& & &\minus & & & & &\plus & \\
& &\minus & & & & &\plus & & \\
&\minus & & & & &\plus & & & \\
\minus & & & & &\plus & & & & \\
\minus & & &\minus &\plus &\minus & & &\minus &\plus \\
&\minus &\minus & &\plus & &\minus &\minus & &\plus \\
\minus & & &\plus & &\minus & & &\plus & \\
&\minus &\plus & & & &\minus &\plus & & \\
\end{tabular}\right\rgroup$\\
	\\[1.75em]
	\hfill $M_{3,2,2}=$ & \\
\end{tabular}
\\[1em]

\end{document}